\documentclass{article}
\usepackage{setspace}
\usepackage{PRIMEarxiv}
\usepackage[utf8]{inputenc} 
\usepackage[T1]{fontenc}    
\usepackage{url}            
\usepackage{booktabs}       
\usepackage{amsfonts}       
\usepackage{nicefrac}       
\usepackage{microtype}      
\usepackage{lipsum}
\usepackage{fancyhdr}       
\usepackage{graphicx}       
\usepackage{bm}
\usepackage{multirow}
\usepackage{adjustbox}
\usepackage{float}
\usepackage{parskip}
\usepackage{algorithm,algpseudocode}
\usepackage{amsthm}
\usepackage{soul}
\usepackage{comment}
\usepackage{subcaption}

\newtheorem{theorem}{Theorem}[section]

\newtheorem{corollary}[theorem]{Corollary}

\newcommand{\bs}{\mathbf{s}}

\graphicspath{{media/}}     
\RequirePackage{amsthm,amsmath,amsfonts,amssymb}
\RequirePackage[round,authoryear]{natbib}
\RequirePackage[colorlinks,citecolor=blue,urlcolor=blue,linkcolor=blue]{hyperref}

\title{Making Recursive Bayesian Inference Robust}

\author{
  Myungsoo Yoo\thanks{Author of correspondence. email: myungsoo.yoo@austin.utexas.edu} \\
  The University of Texas at Austin \\
   \And
  Daniel W\"urzler Barreto\\
  The University of Texas at Austin \\
   \And
  Mevin B. Hooten\\
  The University of Texas at Austin \\
}

\begin{document}
\maketitle

\begin{abstract}
While Bayesian inference has become increasingly popular with advances in computational resources, its algorithms can be computationally prohibitive and may not scale with large datasets. This has led to growing interest in alternative algorithms, such as approximation methods and variants of Markov chain Monte Carlo. Among these approaches, prior proposal–recursive Bayesian (PP-RB) inference facilitates scalable Bayesian computation by recursively updating the posterior distribution across stages and utilizing parallel computing resources. While the well-known ``degeneracy'' issue in PP-RB has been studied, another limitation that PP-RB can yield incorrect inferences when posterior distributions shift substantially between stages has remained unsolved. To address this, we propose parallel-tempered prior proposal-recursive Bayesian (PPP-RB) inference, which extends PP-RB by leveraging the key idea underlying Metropolis-coupled Markov chain Monte Carlo. We show both theoretically and empirically that PPP-RB targets the true posterior distribution. We illustrate PPP-RB through numerical studies and real data analysis in application to earthquake count data and sea surface salinity in the North Atlantic region. In these applications, we compare PPP-RB with PP-RB and a standard MCMC, demonstrating that PPP-RB is more efficient in terms of effective sample size per elapsed time.
\end{abstract}



\section{Introduction}\label{intro}
With the advancement of computational resources, Bayesian methods have gained considerable attention \citep{gelman1995bayesian}. In particular, Bayesian hierarchical models \citep[BHMs;][]{berliner1996hierarchical}, which represent each substage (i.e., data, process, and parameter models) via conditional distributions, offer substantial flexibility for modeling dependent data and have become one of the most widely used modeling frameworks \citep{wikle1998hierarchical, wikle2003hierarchical,dunson2009bayesian, johnson2022greater}. Despite their flexibility, computational methods such as Markov Chain Monte Carlo \citep[MCMC;][]{gelfand1990sampling} often struggle to scale with large datasets and can face significant challenges in terms of mixing and convergence \citep{ green2015bayesian,robert2018accelerating,nemeth2021stochastic}. 

One approach to facilitating Bayesian inference for large datasets is the use of approximation methods, which provide scalable alternatives to traditional MCMC algorithms. These include variational inference \citep[VI;][]{jordan1999introduction}, integrated nested Laplace approximations \citep[INLA;][]{rue2009approximate}, and neural posterior estimation \citep{papamakarios2016fast}, among others. However, these approaches only allow us to approximate the posterior distribution. For example, VI often underestimates posterior variance \citep{blei2017variational}. Rather than approximating the target posterior distribution, certain methods are designed to preserve the correct target distribution while improving computational efficiency. Examples include delayed acceptance MCMC \citep{christen2005markov}, pseudo-marginal MCMC \citep{andrieu2009}, stochastic gradient MCMC \citep{welling2011bayesian}, and subsampling-based MCMC \citep{bradley2021approach,saha2025incorporating}.

Among these methods, prior proposal–recursive Bayesian (PP-RB) inference \citep{hooten2021making} provides a principled framework for scalable Bayesian computation that leverages parallel computing resources while preserving the exact target posterior distribution. The key idea is to partition the data and update the posterior distribution in multiple stages by sequentially incorporating subsets of the data using a standard Metropolis–Hastings (MH) or importance sampling (IS) algorithm \citep{metropolis1953equation,hastings1970monte,kloek1978bayesian}. Notably, posterior samples from the previous stage are randomly selected as proposals for the MH algorithm, thereby allowing the precomputation of the MH acceptance ratio between stages using parallel computing resources and substantially improving the computational efficiency. 
PP-RB and its extensions have been utilized across diverse applications, including ecology \citep{mccaslin2021hierarchical,feuka2022individual,ren2026multi}, optimal design settings in ecological science \citep{leach2022recursive}, hydrology \citep{hepler2025two}, and environmental science \citep{hooten2021making, barreto2025recursive}.

Similar to IS and sequential Monte Carlo, PP-RB can suffer from sample ``degeneracy,'' whereby the effective sample size decreases substantially in subsequent stages. This occurs because only a subset of proposals from the previous stage are accepted in the MH algorithm at the current stage, and these can be further rejected in subsequent stages, thereby reducing the number of available proposals. Several approaches have been proposed to address the degeneracy issue in PP-RB and its extensions. For example, \citet{taylor2025generative} extended PP-RB by applying the transition kernel to posterior samples from the previous stage. \citet{scharf2025strategy} suggested drawing proposals from a smoothed continuous distribution that closely approximated the posterior distribution from the previous stage. In the framework of importance sampling combined with PP-RB, \citet{barreto2025recursive} proposed an algorithm that recursively updates importance sampling weights and introduced sample replenishment steps at an optimal rate.

Beyond the degeneracy issue, PP-RB inherently assumes that the posterior distributions at each stage are located near the final posterior, meaning that the posterior at the current stage does not shift substantially from that at the next stage. However, the posterior distribution can shift in practice, especially when newly introduced data (or partitions) deviate from the previous partition. In such cases, proposals from the previous stage are no longer representative of the current posterior. This posterior shift is distinct from, yet related to, the degeneracy issue because posterior shifts can cause degeneracy but not all degeneracy is caused by posterior shifts. This aspect of PP-RB has received little attention and remains largely unaddressed.

To address this limitation and enable more robust posterior inference, we propose a new framework: the Parallel-tempered Prior Proposal–Recursive Bayesian (PPP-RB) inference framework. PPP-RB extends the PP-RB framework by incorporating a key idea from Metropolis-coupled Markov chain Monte Carlo \citep[MCMCMC;][]{geyer1991mcmc}. MCMCMC enables exploration of a broader parameter space by tempering the posterior distribution. It is particularly effective for exploring multimodal posterior distributions and can improve mixing \citep{gilks1996strategies}. We leverage this strength to address posterior shifts across stages. 

\begin{figure}[t]
\centering
\includegraphics[width=1\linewidth, trim=0 10 0 0, clip]{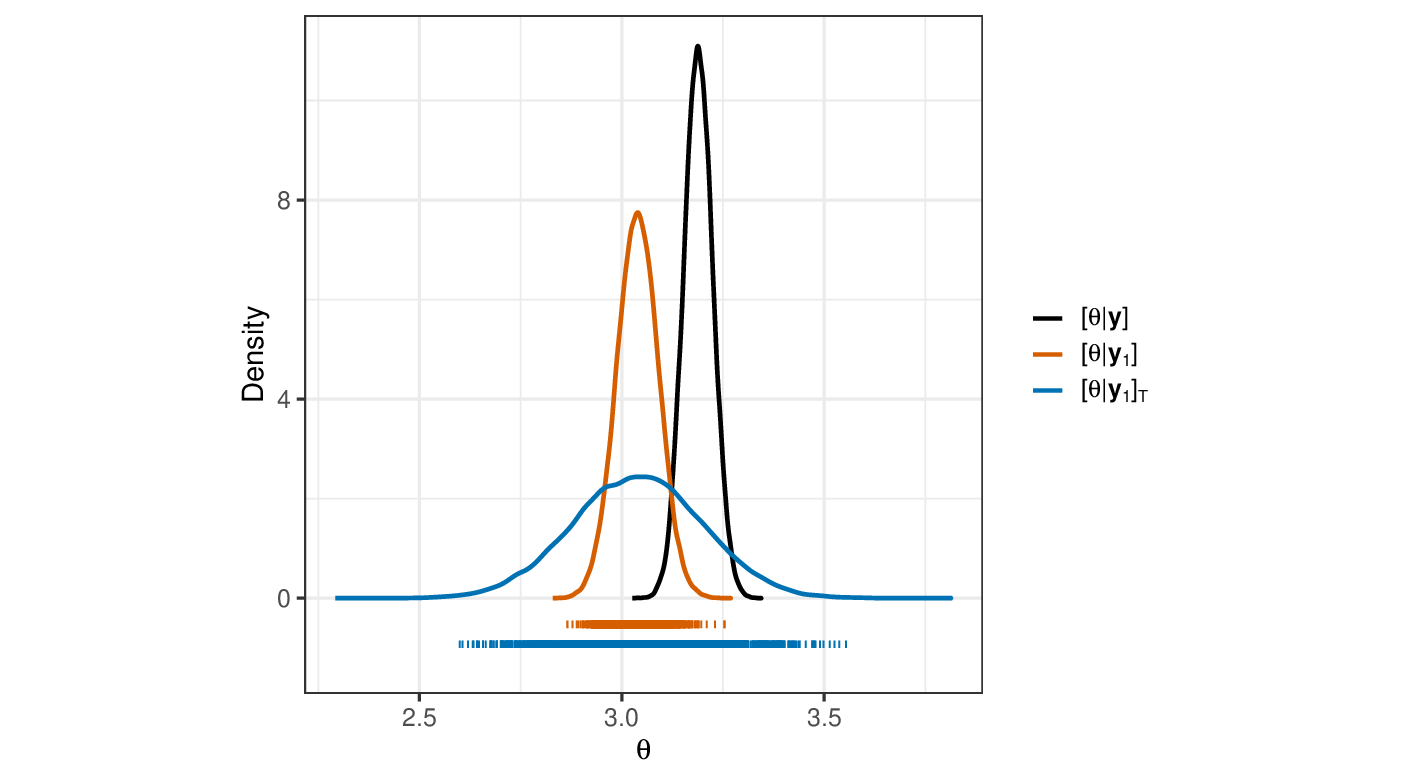}
\caption{Effect of tempering on the posterior distribution. Black, orange, and blue curves show $[\theta \mid\bm{y}]$, $[\theta \mid\bm{y}_1]$, and the tempered posterior $[\theta \mid\bm{y}_1]_{\tau} \propto [\bm{y}_1 \mid\theta]^{1/\tau}[\theta]$, respectively. Tick marks below the density indicate 2,000 proposals randomly sampled from $[\theta \mid\bm{y}_1]$ and $[\theta \mid\bm{y}_1]_{\tau}$ respectively, for illustration purposes.}
\label{teaser}
\end{figure}

To illustrate the effect of tempering, consider the model $[\bm{y}\mid\theta]$, where $``[\cdot]''$ and $``[\cdot \mid\cdot]''$ denote the probability distribution and conditional probability distribution, respectively \citep{gelfand1990sampling}, with data $\bm{y}=(\bm{y}_1',\bm{y}_2')^{'}$ and parameter $\theta$. In the first stage of PP-RB, fitting the model to $\bm{y}_1$ alone yields $[\theta \mid\bm{y}_1]$. When this deviates substantially from the full posterior $[\theta \mid\bm{y}]$, most proposals drawn from $[\theta \mid \bm{y}_1]$ are poor proposals for the MH algorithm at the next stage. In contrast, the tempered posterior $[\theta \mid\bm{y}_1]_{\tau} \propto [\bm{y}_1 \mid \theta]^{1/\tau}[\theta]$ alleviates this by flattening the distribution, producing more representative proposals (Figure \ref{teaser}).

The main contributions of this work are as follows. First, PPP-RB enables more robust inference by utilizing MCMCMC. Compared to PP-RB, PPP-RB explores a broader region of the parameter space, allowing proposals to better adapt to changes in the posterior distribution across stages. Second, we empirically demonstrate that PPP-RB is computationally more efficient than the MH algorithm and yields a greater effective sample size per unit time. Third, we theoretically show that PPP-RB targets the correct posterior distribution. We compare PPP-RB, PP-RB, and the MH algorithm using simulated data and two real datasets: earthquake counts and sea surface salinity.


\section{Background}
\label{background}
\subsection{Prior Proposal-Recursive Bayesian Inference}
\label{pprb}
Instead of fitting a model to a full dataset all at once, recursive Bayesian inference, which is also known as Bayesian filtering or recursive Bayes, involves fitting a model in a series of steps \citep{sarkka2023bayesian}. In that sense, recursive Bayesian inference may be better understood in settings where data are collected over time (i.e., in an online setting). To see this, consider a dataset partitioned into $\bm{y}_1$ and $\bm{y}_2$ (i.e., $\bm{y}=(\bm{y}_1',\bm{y}_2')^{'}$) with associated parameter $\bm{\theta}$ in the model for $\bm{y}$. By Bayes' theorem, the posterior distribution of $\bm{\theta}$ conditioned on $\bm{y}$ is:
\begin{align}
\label{bayes_prior}
[\bm{\theta} \mid \bm{y}_1,\bm{y}_2] &\propto [\bm{y}_1,\bm{y}_2 \mid \bm{\theta}] [\bm{\theta}] \nonumber \\
&\propto [\bm{y}_2 \mid \bm{y}_1 ,\bm{\theta}] [\bm{\theta} \mid\bm{y}_1 ], \end{align}
where $[\bm{\theta}]$ and $[\bm{\theta} \mid \bm{y}_1]$ denote the prior distribution of $\bm{\theta}$ and posterior distributions of $\bm{\theta}$ conditioned on $\bm{y}_1$. The posterior distribution of $\bm{\theta}$ can thus be factored into two components: the likelihood of $\bm{y}_2$ conditioned on $\bm{y}_1$ and $\bm{\theta}$, and the posterior distribution of $\bm{\theta}$ conditioned only on $\bm{y}_1$. This indicates that the posterior distribution $[\bm{\theta} \mid\bm{y}_1,\bm{y}_2]$ can be derived recursively by first fitting the model to $\bm{y}_1$ to obtain $[\bm{\theta} \mid\bm{y}_1]$ and then updating this distribution using the information contained in $\bm{y}_2$. Therefore, recursive Bayesian inference eliminates the need to refit the model to the full dataset $\bm{y}$ when new data $\bm{y}_2$ become available, enabling Bayesian inference to be performed sequentially as data arrive in an online setting and making the algorithm more computationally efficient and scalable.

Recursive Bayesian inference is not limited to online settings. In fact, a benefit of recursive Bayesian inference lies in its efficiency and scalability, which are often required for large datasets. One such approach is PP-RB \citep{hooten2021making}, which is based on Bayes' theorem in Equation \eqref{bayes_prior} together with the two-stage algorithm of \citet{lunn2013fully}. In this framework, the posterior samples obtained in the previous stage are recursively used in the current stage. For illustration, consider a dataset that is partitioned into $J$ subsets, $\bm{y}=(\bm{y}_1',...,\bm{y}_J')^{'}$. Given $[\bm{\theta} \mid \bm{y}_1]$, obtained by fitting the model to $\bm{y}_1$, PP-RB updates the posterior distribution recursively over the remaining $(J-1)$ stages. As in Equation \eqref{bayes_prior}, the posterior distribution of $\bm{\theta}$ at the $j$th stage is updated as
\[
[\bm{\theta} \mid\bm{y}_{1:j}] \propto [\bm{y}_j \mid \bm{\theta} , \bm{y}_{1:(j-1)}] [ \bm{ \theta} \mid \bm{y}_{1:(j-1)}],
\]
where $\bm{y}_{1:(j-1)}=(\bm{y}_1',...,\bm{y}_{j-1}')^{'}$. To update $\bm{\theta}$ at the $k$th iteration in the $j$th stage, the MH algorithm is used with acceptance ratio given by $\min(1,r_j)$, where
\begin{align}
r_j&=\frac{[ \bm{y}_j \mid\bm{\theta}^*, \bm{y}_{1:(j-1)}][\bm{\theta}^*\mid \bm{y}_{1:(j-1)}] [\bm{\theta}^{(k-1)}\mid \bm{y}_{1:(j-1)}] }{ [ \bm{y}_j \mid\bm{\theta}^{(k-1)}, \bm{y}_{1:(j-1)}][\bm{\theta}^{(k-1)}\mid \bm{y}_{1:(j-1)}] [\bm{\theta}^{*}\mid \bm{y}_{1:(j-1)}]} \nonumber \\
&=\frac{[ \bm{y}_j \mid\bm{\theta}^*, \bm{y}_{1:(j-1)}] }{ [ \bm{y}_j \mid\bm{\theta}^{(k-1)}, \bm{y}_{1:(j-1)}]} \label{pprb_mh}.
\end{align}
We let $\bm{\theta}^{(k-1)}$ and $\bm{\theta}^*$ denote the posterior sample of $\bm{\theta}$ at the $(k-1)$th iteration and a proposal drawn from $[\bm{\theta} \mid\bm{y}_{1:(j-1)}]$, respectively. Note that the posterior distribution obtained at the $(j-1)$th stage is used as the proposal distribution in the MH algorithm, which simplifies $r_j$. In practice, posterior samples from $[\bm{\theta} \mid\bm{y}_{1:(j-1)}]$ are randomly drawn with replacement and used as proposals \citep{lunn2013fully, hooten2021making}. Consequently, $r_j$ can be precomputed for each $\bm{\theta}^*$ from $[\bm{\theta} \mid \bm{y}_{1:(j-1)}]$, making the MH algorithm at the subsequent stages efficient. Importantly, parallel computational resources can be leveraged to further improve the efficiency of the algorithm, making PP-RB useful when fitting a model to the full dataset is computationally infeasible. For example, in large spatial datasets, standard Gaussian process models become infeasible due to their $\mathcal{O}(N^3)$ complexity, where $N$ denotes the number of observations. In contrast, PP-RB reduces the cost to $\mathcal{O}(n_1^3)+\frac{1}{R}\sum_{j=2}^{J} \mathcal{O}(n_j^3)$ by operating on data partitions, where $n_j$ denotes the number of observations in the $j$th partition and $R$ denotes the number of parallel computing cores.


\subsection{Power-Tempered PP-RB}
\label{tempered_pprb}
One possible approach to addressing the limitation of PP-RB (i.e., posterior shifts across stages) is to ``temper'' the posterior distribution from the previous stage so that the proposals for the current stage explore a broader region of the parameter space. For illustration, consider two stage PP-RB. The posterior distribution in the first stage is tempered as $[\bm{\theta} \mid \bm{y}_{1}]_{\tau} \propto [\bm{y}_1 \mid \bm{\theta}]^{1/\tau} [\bm{\theta}]$, where $\tau>1$ is the temperature parameter. In the second stage, proposals from $[\bm{\theta} \mid\bm{y}_1]_{\tau}$ at the $k$th iteration are accepted according to the MH acceptance ratio given by $\min (1,r)$, where 
\[
r=\frac{ [\bm{y}_2 \mid \bm{y}_1 ,\bm{\theta}^*] [ \bm{y}_1 \mid \bm{\theta}^*]^{1-1/\tau} }{ [\bm{y}_2 \mid \bm{y}_1 ,\bm{\theta}^{(k-1)}] [ \bm{y}_1 \mid \bm{\theta}^{(k-1)}]^{1-1/\tau} }.
\] 

It is important to select an optimal temperature ${\tau}^*$ because $\tau$ controls the geometry of the tempered posterior distribution. One criterion for selecting $\tau^*$ is the $\chi^2-$divergence between $[\bm{\theta} \mid \bm{y}_{1:J}]$ and $[\bm{\theta} \mid \bm{y}_{1}]_{\tau}$, given as
\begin{align}
\label{chi_divege}
\tau^* = \arg \min_{\tau} D_{\chi^2} \left( [\bm{\theta} \mid \bm{y}_{1:J}] \,\middle\|\, [\bm{\theta} \mid \bm{y}_{1}]_{\tau} \right),
\end{align}
where $ D_{\chi^2} \left( [\bm{\theta} \mid \bm{y}_{1:J}] \,\middle\|\, [\bm{\theta} \mid \bm{y}_{1}]_{\tau} \right)=\int
\frac{
\left\{ [\bm{\theta} \mid \bm{y}_{1:J}] \right\}^{2}
}{
[\bm{\theta} \mid \bm{y}_{1}]_{\tau}
}
\, d\bm{\theta}
- 1$. This criterion is also related to finding $\tau^*$ such that it maximizes a rough approximation of the effective sample size \citep{liu1995blind,martino2017effective}. In other words, by increasing the ESS, the power-tempered PP-RB enables robust inference under posterior shifts.

For illustration, suppose the data are partitioned into $J$ subsets where  $\bm{y}_{ij} \mid \bm{\theta}\sim \text{Gau}(\bm{\theta},\bm{\Sigma})$ and $\bm{\theta}\sim \text{Gau}(\bm{\theta_0,\bm{\Sigma}_0})$ are assumed for $\bm{y}_{ij}\in \mathbb{R}^d$ with $i\in \{1,...,n_j\}$ and $j \in \{1,...,J\}$. Assuming that $\bm{\Sigma}$ is known, posterior distributions, $[\bm{\theta} \mid \bm{y}_1]_T =\text{Gau}(\bm{\theta}_{1,T},\bm{\Sigma}_{1,T})$ and $[\bm{\theta} \mid \bm{y}_{1:J}] =\text{Gau}(\bm{\theta}_{J},\bm{\Sigma}_{J})$, can be analytically obtained, where $\bm{\Sigma}_J=\bm{\Sigma}_0+n \bm{\Sigma}$, $\bm{\theta}_J = \bm{\Sigma}_J^{-1}(\bm{\Sigma}_0 \bm{\theta}_0 +n \bm{\Sigma} \bar{\bm{y}}_J)$, $\bm{\Sigma}_{1,T}=\bm{\Sigma}_0 + \frac{n_1}{T} \bm{\Sigma}$, and $\bm{\theta}_{1,T}= \bm{\Sigma}^{-1}_{1,T} (\bm{\Sigma}_0 \bm{\theta}_0+\frac{n_1}{T} \bm{\Sigma} \bar{\bm{y}}_1)$ where $n=\sum_{j=1}^J n_j$, $\bar{\bm{y}}_J = \frac{1}{n}\sum_{j=1}^{J}\sum_{i=1}^{n_j} \bm{y}_{ij}$, and $\bar{\bm{y}}_1 = \frac{1}{n_1} \sum_{i=1}^{n_1} \bm{y}_{i1}$. In this case, it can be analytically shown that 
\begin{align}
\label{optimal_t_star}
\tau^*= \frac{n_1}{n}\bigg \{\frac{3}{2} + \frac{1}{d} \text{tr}(\bm{S})- \frac{1}{2}\sqrt{1+4 \bigg(\big(\frac{1}{d}\text{tr}(\bm{S})+\frac{3}{2}\big)^2-\big(\frac{3}{2}\big)^2\bigg)} \bigg\}^{-1}.
\end{align}
Here, $\bm{S}= n \bm{\Sigma}(\bar{\bm{y}}_J-\bar{\bm{y}}_1)(\bar{\bm{y}}_J-\bar{\bm{y}}_1)'.$ Figure \ref{ESS} compares the empirical ESS values obtained using PP-RB with the theoretical ESS derived using the delta method. It can be seen that $\tau^* \neq 1$, where $\tau^*$ is obtained by Equation~\eqref{optimal_t_star}, yields the highest ESS, illustrating the advantage of tempering over the untempered case ($\tau = 1$). See the supplemental material for details, including the theoretical ESS derived using the delta method.

\begin{figure}[t]
\centering
\includegraphics[width=1\linewidth]{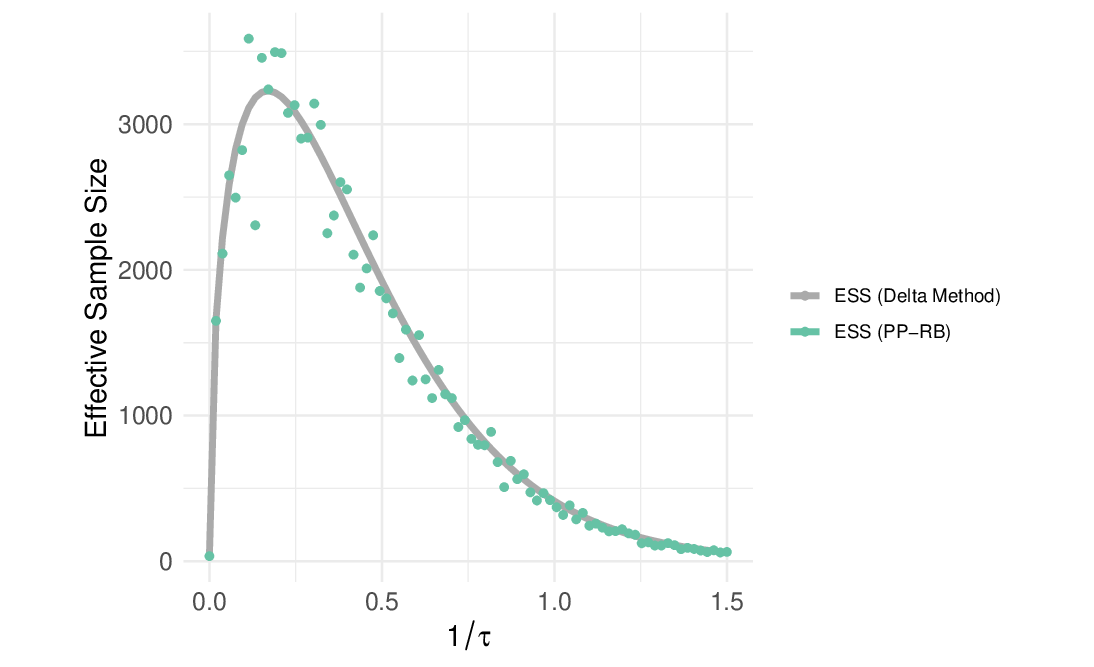}
\caption{ESS estimated from the PP-RB for each $1/\tau \in (0,1.5)$. The solid line shows the theoretical ESS obtained via the delta method.}
\label{ESS}
\end{figure}

While this example demonstrates the effect of tempering in addressing posterior shifts across stages, it is noteworthy that Equation \eqref{chi_divege} depends on $[\bm{\theta} \mid \bm{y}_{1:J}]$, which is typically not available except in certain special cases. This limitation motivates our proposed algorithm, PPP-RB, which is based on the MCMCMC framework described in Section \ref{subsection_tempering}.

\subsection{MCMCMC}
\label{subsection_tempering}
MCMCMC originated as the ``replica exchange'' simulation method \citep{swendsen1986replica}, later formalized by \citet{geyer1991mcmc}, and is now widely known as parallel tempering \citep{earl2005parallel}. MCMCMC can enhance mixing and is particularly useful for multimodal target distributions because it explores the parameter space more broadly \citep{geyer1991mcmc}. MCMCMC has been popular in Bayesian phylogenetic inference \citep[see, e.g.,][]{huelsenbeck2001mrbayes, altekar2004parallel, ronquist2012mrbayes, muller2020adaptive}. For a review of the use of parallel tempering as a simulation method across a variety of applications, see \citet{earl2005parallel}.

The premise of MCMCMC is to construct multiple tempered distributions by raising the target distribution to powers corresponding to different temperatures, which can be written as
\begin{align}
\label{tempered}
[\bm{\theta}\mid \bm{y}]_{\tau_{\ell}} \propto [\bm{y}\mid\bm{\theta}]^{1/ \tau_{\ell}} [\bm{\theta}], \quad \ell=1,...,L,
\end{align}
where $[\bm{\theta} \mid \bm{y}]_{\tau_{\ell}}$ denotes the tempered posterior distribution at temperature $\tau_{\ell}$, $[\bm{y} \mid \bm{\theta}]^{1/\tau_{\ell}}$ is the tempered likelihood, and $\tau_{\ell} \in [1, \tau_{\max}]$ is the temperature parameter. By raising the likelihood to a power corresponding to the temperature, the likelihood surface becomes flatter, thus facilitating exploration of the parameter space and movement between modes. Note that Equation \eqref{tempered} represents the original posterior distribution of $\bm{\theta}$ with $\tau_{\ell} = 1$. Among the $L$ tempered posterior distributions (i.e., $L$ chains), the one with $\tau_{\ell}=1$ is referred to as the ``cold'' chain, while those with $\tau_{\ell}>1$ are referred to as the ``hot'' chains. Without loss of generality, we hereafter assume that the first chain is the cold chain ($\tau_{1}=1$), while the remaining chains are hot chains ($1<\tau_{\ell}<\tau_{\ell+1}<\cdots <\tau_{L}$). With $\tau_{\ell}=1$, the cold chain is designed to be the Markov chain that targets the posterior distribution $[\bm{\theta} \mid\bm{y}]$. In contrast, hot chains play the role of exploring the parameter space more broadly than the cold chain and facilitating movement between local modes in the cold chain via state exchanges (or swaps) between hot and cold chains, thereby improving overall mixing.

More formally, MCMCMC consists of two main steps: (1) drawing posterior samples from $[\bm{\theta} \mid \bm{y}]_{\tau_{\ell}}$ for $\ell=1,\ldots,L$, and (2) exchanging posterior samples between the cold chain and one of the hot chains. Specifically, we draw posterior samples using an MCMC algorithm at each iteration for every chain, and at every $m$th iteration, a swap between the cold chain and the $\ell$th hot chain is proposed and accepted according to the MH acceptance ratio given by $\min(1, r_{\ell}^{\text{swap}})$, where
\begin{align*}
r_{\ell}^{\text{swap}}&=\frac{[\bm{y}\mid \bm{\theta}^{(m)}_{\text{hot}}][\bm{\theta}^{(m)}_{\text{hot}}]\cdot[\bm{y}\mid \bm{\theta}^{(m)}_{\text{cold}}]^{1/\tau_{\ell}}[\bm{\theta}^{(m)}_{\text{cold}}]}{[\bm{y}\mid \bm{\theta}^{(m)}_{\text{cold}}][\bm{\theta}^{(m)}_{\text{cold}}]\cdot[\bm{y}\mid \bm{\theta}^{(m)}_{\text{hot}}]^{1/\tau_{\ell}}[\bm{\theta}^{(m)}_{\text{hot}}]} \nonumber \\
&=\frac{[\bm{y}\mid \bm{\theta}^{(m)}_{\text{hot}}]\cdot[\bm{y}\mid \bm{\theta}^{(m)}_{\text{cold}}]^{1/\tau_{\ell}}}{[\bm{y}\mid \bm{\theta}^{(m)}_{\text{cold}}]\cdot[\bm{y}\mid \bm{\theta}^{(m)}_{\text{hot}}]^{1/\tau_{\ell}}}.
\end{align*}
The terms $\bm{\theta}_{\text{cold}}^{(m)}$ and $\bm{\theta}_{\text{hot}}^{(m)}$ denote the samples from the cold and hot chains at the $m$th MCMC iteration. By exchanging samples, the cold chain is more likely to move between local modes and escape local trapping, improving exploration of the posterior distribution. Notably, a drawback is that MCMCMC requires $L$ chains, but only one (i.e., the cold chain) is ultimately used for inference because the hot chains facilitate exploration rather than inference \citep{gilks1996strategies}. Nevertheless, recent advances in computing resources, such as parallel computing, have made MCMCMC more accessible and useful.

MCMCMC requires tuning. For example, one must determine the maximum temperature $\tau_{\text{max}}$, the number of chains $L$, the temperature schedule (or configuration) $\{\tau_{\ell}\}$ for ${\ell=1,...,L}$, and the swap frequency $m$ for exchanging posterior samples between chains. For the temperature schedule, the geometric temperature schedule (i.e., $\tau_{\ell}=\tau_0\cdot r^{\ell}$, $r>0$) is often used \citep{vousden2016dynamic}. However, more advanced approaches aim to determine an optimal configuration. For example, \citet{atchade2011towards} showed that the optimal temperature schedule should be determined such that the swap ratio between adjacent chains is 0.234. To achieve this target, various adaptive sampling frameworks have been developed, including the algorithms proposed by \citet{miasojedow2013adaptive}, \citet{vousden2016dynamic}, \citet{muller2020adaptive}, and \citet{zhao2024policy}. In their proposed algorithm, the swap acceptance ratio plays a key role, and they aim to adjust the temperature schedule to maintain swap acceptance ratio near 0.234. In practice, maintaining a swap acceptance rate between 0.2 and 0.4 is often used as a rule of thumb to ensure adequate mixing between chains. For a review, see \citet{earl2005parallel}.

While determining the optimal temperature schedule has been widely studied, the number of chains $L$, maximum temperature $\tau_{\text{max}}$, and swap frequency $m$ have received relatively less attention because these are difficult to select without sufficient prior knowledge of $[\bm{\theta} \mid \bm{y}]$. In practice, a large $\tau_{\text{max}}$ is typically recommended to ensure sufficient exploration of the target distribution \citep{earl2005parallel}, and $L$ is often tuned to achieve a swap acceptance rate close to 0.2. The swap frequency $m$ appears to be the least studied. More frequent swaps are generally preferable when $[\bm{\theta} \mid \bm{y}]$ exhibits strong multimodality, though at an increased computational cost. However, because the geometry of $[\bm{\theta} \mid \bm{y}]$ is generally unknown prior to model fitting, selecting $m$ in advance is an ongoing subject of research.

\section{PPP-RB}
\label{ppprb}
As introduced in Section \ref{intro}, PP-RB works well when the posterior distributions do not change substantially between stages. However, this condition is not always satisfied, and PP-RB may result in incorrect inference. By combining the key idea in PP-RB and MCMCMC, we propose PPP-RB which addresses the limitation of PP-RB and facilitates more robust inference. We first present the algorithm in Section \ref{ppprb_algo}, followed by theoretical results in Section \ref{ppprb_theory}, and practical considerations in Section \ref{ppprb_practical}.

\subsection{The algorithm}
\label{ppprb_algo}
As with PP-RB, PPP-RB assumes that the data are partitioned into $J$ subsets $\bm{y}=(\bm{y}_1',...,\bm{y}_J')^{'}$, and updates the posterior distribution recursively by fitting the model to each partition. While PP-RB uses a single MCMC chain, PPP-RB requires multiple chains, as in MCMCMC. PPP-RB can be viewed as a generalization of PP-RB, in the sense that it reduces to PP-RB when only a single chain is used and no swapping is performed. 

In the first stage, we fit the model to $\bm{y}_1$, as in PP-RB, but repeat this for each chain using different temperatures in parallel:
$[\bm{\theta} \mid \bm{y}_1]_{\tau_{\ell}}\propto[\bm{y}_1 \mid\bm{\theta}]^{1/\tau_{\ell}}[\bm{\theta}]$ for $\ell=1,...,L$. Because multiple chains are used, the computational cost for PPP-RB is more expensive than PP-RB. Nonetheless, obtaining multiple chains in the first stage can be done in parallel, resulting in computational time comparable to that of PP-RB.

For the remaining $J-1$ stages, we recursively update the posterior distribution for each chain, where each stage consists of two steps: 1) updating the posterior distribution of $\bm{\theta}$ using the data $\bm{y}_j$, and 2) swapping posterior samples between cold and hot chain. We term these steps ``within-chain update'' and ``between-chain exchange,'' respectively. For within-chain updates, we use the Bayes' theorem as in PP-RB:
\[
[\bm{\theta} \mid \bm{y}_{1:j}]_{\tau_{\ell}} \propto [\bm{y}_{j} \mid \bm{\theta}, \bm{y}_{1:(j-1)}]^{1/\tau_{\ell}} [\bm{\theta}\mid \bm{y}_{1:(j-1)}]_{\tau_{\ell}}, \quad j=2,...,J
\]
where $[\bm{y}_{j} \mid \bm{\theta},\bm{y}_{1:(j-1)}]^{1/\tau_{\ell}}$ denotes the tempered conditional likelihood of $\bm{y}_j$ conditioned on $\bm{\theta}$ and $\bm{y}_1,...,\bm{y}_{(j-1)}$. The MH acceptance ratio for the ${\ell}$th chain at the $j$th stage and $k$th MCMC iteration is given by $\min(1,r_{j,{\ell}})$, where
\begin{align}
r_{j,\ell}&=\frac{ [\bm{y}_j \mid \bm{\theta}^*, \bm{y}_{1:(j-1)}]^{1/\tau_{\ell}} [\bm{\theta}^*\mid\bm{y}_{1:(j-1)}]_{\tau_{\ell}} [\bm{\theta}^{(k-1)}\mid \bm{y}_{1:(j-1)}]_{\tau_{\ell}} }{ [\bm{y}_j \mid \bm{\theta}^{(k-1)}, \bm{y}_{1:(j-1)}]^{1/\tau_{\ell}} [\bm{\theta}^{(k-1)} \mid\bm{y}_{1:(j-1)}]_{\tau_{\ell}} [\bm{\theta}^{*}\mid \bm{y}_{1:(j-1)}]_{\tau_{\ell}}  } \nonumber \\
&=\frac{ [\bm{y}_j \mid \bm{\theta}^*, \bm{y}_{1:(j-1)}]^{1/\tau_{\ell}}  }{ [\bm{y}_j \mid \bm{\theta}^{(k-1)}, \bm{y}_{1:(j-1)}]^{1/\tau_{\ell}} } \label{ppprb_mh}.
\end{align}
In Equation \eqref{ppprb_mh}, $\bm{\theta}^*$ is a proposal randomly drawn (with replacement) from the posterior distribution of the previous stage. Note that the only difference between Equations \eqref{pprb_mh} and \eqref{ppprb_mh} is the temperature $\tau_{\ell}$ applied to $[\bm{y}_j \mid \bm{\theta}^{(k-1)}, \bm{y}_{1:(j-1)}]$. As in PP-RB, parallel computational resources can be used to precompute \eqref{ppprb_mh} for all $\bm{\theta}^*$ in the previous stage, thereby improving computational efficiency. 

After completing the within-chain update, PPP-RB performs the between-chain exchange every $m$ MCMC iterations. This step involves selecting a hot chain and determining whether to swap based on the MH acceptance ratio. The MH acceptance ratio for swapping between the 1st (cold) chain and the $\ell$th (hot) chain at the $j$th stage and $m$th MCMC iteration is given by $\min(1,r_{j,\ell}^{\text{swap}})$, where
\begin{align}
r_{j,\ell}^{\text{swap}}&=\frac{ [\bm{y}_j \mid \bm{y}_{1:(j-1)},\bm{\theta}_{\text{hot}}^{(m)}] [\bm{y}_{1:(j-1)} \mid \bm{\theta}_{\text{hot}}^{(m)}] [\bm{\theta}_{\text{hot}}^{(m)}]\cdot [\bm{y}_j \mid\bm{y}_{1:(j-1)},\bm{\theta}_{\text{cold}}^{(m)}]^{1/\tau_{\ell}} [\bm{y}_{1:(j-1)}\mid\bm{\theta}_{\text{cold}}^{(m)}]^{1/\tau_{\ell}} [\bm{\theta}_{\text{cold}}^{(m)}]}{   [\bm{y}_j \mid \bm{y}_{1:(j-1)},\bm{\theta}_{\text{cold}}^{(m)}] [\bm{y}_{1:(j-1)}\mid \bm{\theta}_{\text{cold}}^{(m)}] [\bm{\theta}_{\text{cold}}^{(m)}]\cdot [\bm{y}_j\mid\bm{y}_{1:(j-1)},\bm{\theta}_{\text{hot}}^{(m)}]^{1/\tau_{\ell}} [\bm{y}_{1:(j-1)}\mid\bm{\theta}_{\text{hot}}^{(m)}]^{1/\tau_{\ell}} [\bm{\theta}_{\text{hot}}^{(m)}] } \nonumber \\
&=\frac{ [\bm{y}_j \mid \bm{y}_{1:(j-1)},\bm{\theta}_{\text{hot}}]^{1-1/\tau_{\ell}} [\bm{y}_{1:(j-1)}\mid \bm{\theta}_{\text{hot}}^{(m)}]^{1-1/\tau_{\ell}}  }{   [\bm{y}_j \mid \bm{y}_{1:(j-1)},\bm{\theta}_{\text{cold}}^{(m)}]^{1-1/\tau_{\ell}} [\bm{y}_{1:(j-1)}\mid \bm{\theta}_{\text{cold}}^{(m)}]^{1-1/\tau_{\ell}}   } \label{swap_ratio}.
\end{align}
In Equation \eqref{swap_ratio}, $\bm{\theta}_{\text{cold}}^{(m)}$ and $\bm{\theta}_{\text{hot}}^{(m)}$ denote the posterior samples from the cold ($\ell=1$) and hot ($\ell>1$) chains at the $m$th MCMC iteration. Note that Equation \eqref{swap_ratio} can also be efficiently precomputed in parallel, thereby making both the within-chain update and between-chain exchange more efficient.

Compared to the tempered PP-RB in Section \ref{tempered_pprb}, PPP-RB is less sensitive to the choice of temperature because multiple chains at different temperatures explore the parameter space and exchange posterior samples. In contrast, tempered PP-RB is more sensitive to the choice of temperature because it relies on only a single temperature.

\begin{algorithm}[t]
\footnotesize
\begin{spacing}{0.95}
    \caption{Steps to implement the PPP-RB}
    \label{algo:1}
    \begin{algorithmic}[1]
        \State \textbf{Input} Partitioned data ${\bm{y}_1, \ldots, \bm{y}_J}$, temperatures $\tau_\ell$ for $\ell = 1, \ldots, L$ chains, and the number of MCMC iteration $K$.
        \State \textbf{for} $\ell=1:L$ \textbf{do} \label{algo:step2}
        \State \hskip1.0em Draw $K$ posterior samples via MCMC from 
        $[\bm{\theta}\mid \bm{y}_1]_{\tau_{\ell}} \propto [\bm{y}_1 \mid \bm{\theta}]^{1/\tau_{\ell}} [\bm{\theta}].$
        \label{algo:step3}
        \State \textbf{end for} \label{algo:step4}
         \State \textbf{for} $j=2:J$ \textbf{do} \label{algo:step6}
        \State \hskip1.0em Precompute the MH ratios for all samples from $[\bm{\theta}\mid \bm{y}_{1:(j-1)}]_{\tau_{\ell}}$ in Equations \eqref{ppprb_mh} and \eqref{swap_ratio}, for each $\ell = 1, \ldots, L$.
        \label{algo:step5}
        \State \hskip1.0em\textbf{for} $k=1:K$ \textbf{do} \label{algo:step7}
         \State \hskip2.0em\textbf{for} $\ell=1:L$ \textbf{do} \label{algo:step8}
        \State \hskip3.0em Draw posterior sample using $\bm{y}_j$ and the MH algorithm with the ratio in Equation \eqref{ppprb_mh}.
        \label{algo:step9}
         \State \hskip2.0em \textbf{end for} \label{algo:step10}
        \State \hskip2.0em At every $m$ iterations, propose swaps between the cold and hot chains and accept with probability given by Equation \eqref{swap_ratio}.  \label{algo:step11}
        \State \hskip1.0em \textbf{end for} \label{algo:step12}
        \State 
        \textbf{end for} \label{algo:step13}
        \State \textbf{output} posterior samples from the cold chain ($\ell = 1$) of $[\bm{\theta}\mid \bm{y}_{1:J}]$
        \label{algo:step14}
    \end{algorithmic}
    \end{spacing}
\end{algorithm}
Algorithm \ref{algo:1} describes the implementation of PPP-RB. In the first stage, $L$ chains with different temperature $\tau_{\ell}$ are fitted to $\bm{y}_1$ to obtain posterior samples from $[\bm{\theta} \mid\bm{y}_1]_{\tau_{\ell}}$ (Steps \ref{algo:step2}-\ref{algo:step4}). Before proceeding to the next stage, the MH acceptance ratios in Equations \eqref{ppprb_mh} and \eqref{swap_ratio} are precomputed for the next stage using parallel computing resource (Step \ref{algo:step5}). At subsequent stages ($j=2,...,J$), the within-chain update (Step \ref{algo:step9}) is performed, where the proposal $\bm{\theta}^*$ is randomly drawn from the posterior distribution obtained at the previous stage. Then, after selecting a hot chain for between-chain exchange, posterior samples between cold and hot chains are exchanged at every $m$ iterations according to the MH acceptance ratio in Equation \eqref{swap_ratio} (Step \ref{algo:step11}). Within-chain update and between-chain exchange are repeated for the remaining $J-1$ stages to obtain the posterior samples of $[\bm{\theta} \mid \bm{y}_{1:J}]$ from the cold chain. 

\subsection{Theoretical results}
\label{ppprb_theory}
In this section, we show that the cold chain of PPP-RB targets the true posterior distribution by verifying the detailed balance condition. A full proof is provided in the supplemental material.
\begin{corollary}
\label{corollary1}
In the first stage of PPP-RB, each Markov chain with temperatures $\tau_{\ell} \in [1,\tau_{\text{max}}]$ for $\ell=1,...,L$ has the corresponding tempered posterior distribution as its stationary distribution, given by
\[
[\bm{\theta}\mid\bm{y}_1]_{\tau_{\ell}}\propto [\bm{y}_1\mid \bm{\theta}]^{1/\tau_{\ell}} [\bm{\theta}],
\]
where $\bm{y}_1$ denotes the first partition of the data $\bm{y}=(\bm{y}_1',...,\bm{y}_J')^{'}$, $\bm{\theta}$ represents the model parameters, and $\tau_{\text{max}}$ is the maximum temperature used in PPP-RB. 
\end{corollary}
Corollary \ref{corollary1} follows from standard results on the MH algorithm. In particular, the corresponding transition kernel satisfies the detailed balance condition with respect to the tempered posterior distributions, which implies that these distributions are invariant (and hence stationary) for the kernel.

\begin{corollary}
\label{corollary2}
Let $\bm{\Theta}$ denote the augmented parameter space for $L$ Markov chains, and define the joint tempered distribution as
\[
    \Pi(\bm{\Theta}) = \prod_{\ell=1}^{L} [\bm{\theta} \mid\bm{y}]_{\tau_{\ell}},
\]
where $[\bm{\theta} \mid\bm{y}]_{\tau_{\ell}} \propto [\bm{y} \mid\bm{\theta}]^{1/\tau_{\ell}}\,[\bm{\theta}]$.
In the remaining $J-1$ stages of PPP-RB, both the within-chain update and the 
between-chain exchange at each stage satisfy detailed balance with respect to 
$\Pi(\bm{\Theta})$. Consequently, $\Pi(\bm{\Theta})$ is invariant under the 
transition kernel, where the transition kernel denotes the kernel of PPP-RB combining the within-chain update and between-chain exchange steps.
\end{corollary}
Corollary \ref{corollary2} follows from results associated with MCMCMC. 
Detailed balance for the within-chain update holds by the results on the MH algorithm, applied to each $[\bm{\theta} \mid\bm{y}]_{\tau_{\ell}}$ individually. For the between-chain exchange, we consider the joint distribution over all 
chains. A proposal for the between-chain exchange is symmetric and the MH acceptance ratio in Equation \eqref{swap_ratio} ensures that the detailed balance condition holds with respect to $\Pi(\bm{\Theta})$.
\begin{corollary}
\label{corollary3}
After $J-1$ recursive updates of PPP-RB, the marginal stationary 
distribution of the cold chain
\[
\int \cdots\int \Pi(\bm{\Theta})\, d\bm{\theta}_2 \cdots d\bm{\theta}_L 
\]
is the posterior distribution $[\bm{\theta} \mid \bm{y}]$.
\end{corollary}

Corollary~\ref{corollary3} follows directly from 
Corollaries~\ref{corollary1} and~\ref{corollary2}. Because 
$\Pi(\bm{\Theta})$ is invariant under the transition kernel of  PPP-RB, and the cold chain has temperature $\tau_1=1$, its marginal 
reduces to the true posterior $[\bm{\theta} \mid\bm{y}]$.

\subsection{Practical consideration}
\label{ppprb_practical}
Although the implementation of PPP-RB is straightforward, there are several practical considerations that must be addressed in actual implementation. In this section, we provide general guidelines for these considerations.

\subsubsection{Max temperature}
Although PPP-RB can address the limitations of PP-RB introduced in Section \ref{intro}, it is important to ensure that the posterior distributions of the hot chains in the first stage of PPP-RB cover a sufficiently broad parameter space. In other words, if these posterior distributions do not cover the relevant parameter space (i.e., the region of non-negligible density of the posterior distribution obtained when fitting the model to the full data), PPP-RB may still fail to perform adequately. In that sense, it is important to determine the maximum temperature $\tau_{\text{max}}$. As with MCMCMC, it is not straightforward to establish a universal rule for the optimal $\tau_{\text{max}}$ because it may depend on the specific model and it requires the information on $[\bm{\theta} \mid \bm{y}]$. This is further complicated by the choice of partitioning scheme, which can be stochastic (e.g., random partition) and may vary across applications, making it difficult to provide a universally applicable guideline. However, it is worth noting that the optimal $\tau_{\text{max}}$ for PPP-RB may be lower than that for MCMCMC because PPP-RB relies only on a single partition in the first stage, and the likelihood based on a single partition is typically flatter than the full-data likelihood. In that sense, PPP-RB may achieve sufficient exploration of the parameter space with a relatively smaller $\tau_{\text{max}}$ than MCMCMC. One possible approach is to select $\tau_{\text{max}}$ such that the resulting posterior distribution is close to the prior distribution. Although finding such $\tau_{\text{max}}$ may require fitting multiple hot chains at different temperatures, the first stage is not computationally expensive, and parallel computing can further reduce the computational cost.

\subsubsection{Number of chains and temperature schedule}
The number of chains $L$ can be determined in conjunction with $\tau_{\text{max}}$. In other words, a large $L$ is recommended when $\tau_{\text{max}}$
is large, whereas a small $L$ may suffice when $\tau_{\text{max}}$ is small. This is to ensure that the gaps between temperatures are not too large because large gaps can lead to substantial changes in the log-likelihood within the MH acceptance ratio (i.e., Equation \eqref{swap_ratio}), resulting in an unacceptably low acceptance ratio for the between-chain exchange. As with MCMCMC, the acceptance ratio for the between-chain exchange is crucial for efficient exploration of the parameter space. An acceptance ratio between 0.2 and 0.4 tends to indicate that the gaps between adjacent temperatures are appropriate for efficient between-chain exchange. If the acceptance ratio falls below this range, one should increase $L$. Conversely, if the acceptance ratio exceeds this range, $L$ may be decreased.

While various approaches for adaptively tuning the temperature schedule have been proposed in the context of MCMCMC, as discussed in Section \ref{subsection_tempering}, they do not appear to be directly applicable to PPP-RB. This is because the temperature schedule determined in the first stage of PPP-RB is assumed to be fixed for the remaining stages, and adaptively adjusting the temperature schedule would prevent the chains from being obtained in parallel, thereby making the first stage computationally more expensive. Nonetheless, we observed empirically that geometrical spacing performs well in both numerical studies and real data analyses. This may be because geometric spacing provides a balanced coverage of the temperature range, avoiding excessively large gaps at low temperatures while still allowing sufficiently wide exploration at higher temperatures.

\subsubsection{Swap frequency and hot chain selection}
We recommend frequent between-chain exchanges (i.e., $m=1$) to improve the robustness of PPP-RB, especially when the posterior distribution in the current stage is expected to be significantly different from the posterior distribution at the next stage. By setting $m=1$, the algorithm maximizes the opportunity for potential exchanges between the cold and a hot chain, making it more likely that information from the hot chain is propagated to the cold chain and improving mixing. The frequent between-chain exchanges is particularly useful in the presence of uncertainty in $\tau_{\text{max}}$, $L$, and the temperature schedule. For illustration, consider the case where only a few hot chains correctly explore the regions of the parameter space that overlap with the true posterior distribution, due to the suboptimal $\tau_{\text{max}}$, $L$, or temperature schedule. In such cases, frequent swap proposals are crucial to allow the cold chain to ``access'' these informative hot chains more frequently, thereby ensuring that the cold chain can still effectively benefit from the limited number of hot chains that adequately explore the regions of the parameter space of the true posterior distribution. 

One should also select a hot chain for the between-chain exchange. The simplest approach is to select a hot chain at random. A more sophisticated approach is to track the MH acceptance ratio for each hot chain during burn-in and select a hot chain proportional to the acceptance ratio after burn-in. In our application, we randomly select a hot chain and we find this performs well in practice.
\section{Simulation study}
\label{simulation}
In this section, we demonstrate how PPP-RB can enable more robust inference than PP-RB, particularly in simulation settings where the posterior distributions at each stage differ substantially. 
\subsection{A Normal mixture of a light-tailed and a heavy-tailed distribution}
\label{simul_2}
We begin with a simple normal mixture model to assess how well PPP-RB is able to handle posterior shifts across stages. We simulated $N=1,000$ observations as
\[
y_i \sim p\cdot \text{Gau}(\mu, \sigma^2_1)+(1-p)\cdot \text{Gau}(\mu,\sigma^2_2),
\]
where $\mu=0$, $\sigma^2_1=0.3$, $\sigma^2_2=5$, and $p=0.8$. We then fit the same model to the data and compared the posterior distributions of the parameters obtained using the MH algorithm, PP-RB, and PPP-RB. To partition the data into two subsets $\bm{y}_1$ and $\bm{y}_2$, we first computed a threshold $Q_{0.85}(\bm{y})$, defined as the 0.85 quantile of the data. Based on $Q_{0.85}(\bm{y})$, we defined two index sets $\mathcal{I}_1 = \{ y_i : |y_i| \leq Q_{0.85}(\bm{y}) \}$ and $\mathcal{I}_2 = \{ y_i : |y_i| > Q_{0.85}(\bm{y}) \}$. We then constructed $\bm{y}_1$ by randomly selecting 60\% of the observations in $\mathcal{I}_1$ and 30\% of the observations in $\mathcal{I}_2$. Finally, $\bm{y}_2$ consisted of all observations not included in $\bm{y}_1$. As shown in the left panel of Figure \ref{fig:simul1_b}, $\bm{y}_1$ contained more observations from the bulk of the distribution, whereas $\bm{y}_2$ included more observations from the tail.

For all algorithms, we assigned priors $p \sim \text{Beta}(8,2)$, $\mu \sim \text{Gau}(0,100)$, $\sigma_1^2 \sim \text{IG}(15,4.2)$, and $\sigma_2^2 \sim \text{IG}(15,70)$, and we drew 50,000 MCMC samples with the first 5,000 as burn-in. The acceptance ratios were between 0.2 and 0.4. For PPP-RB, we used $L=5$ chains consisting of one cold chain and four hot chains. The temperature $\{T_\ell\}_{\ell=1}^5$ was defined by $T_1=1$, and $T_\ell=\exp(s_\ell)$ for $\ell=2,\ldots,5$, where $s_\ell$ were equally spaced over the interval $[0,2]$.

From Figure \ref{fig:simul1_b}, the first-stage posterior summaries for $p$ under PP-RB and PPP-RB differed substantially from those obtained with the MH algorithm. By contrast, for the remaining parameters, PP-RB and PPP-RB produced posterior summaries that were consistent with the MH results for $\mu$, $\sigma^2_1$ and $\sigma^2_2$. However, due to the discrepancy in $p$,  PP-RB at the second stage still yielded posterior distributions that were significantly different from those obtained using the MH algorithm. By contrast, PPP-RB produced posterior distributions for all parameters that were closely aligned with the MH results.

\begin{figure}[t]
\centering
\begin{subfigure}{1\textwidth}
    \centering
    \includegraphics[width=\linewidth, trim=0 20 0 10, clip]{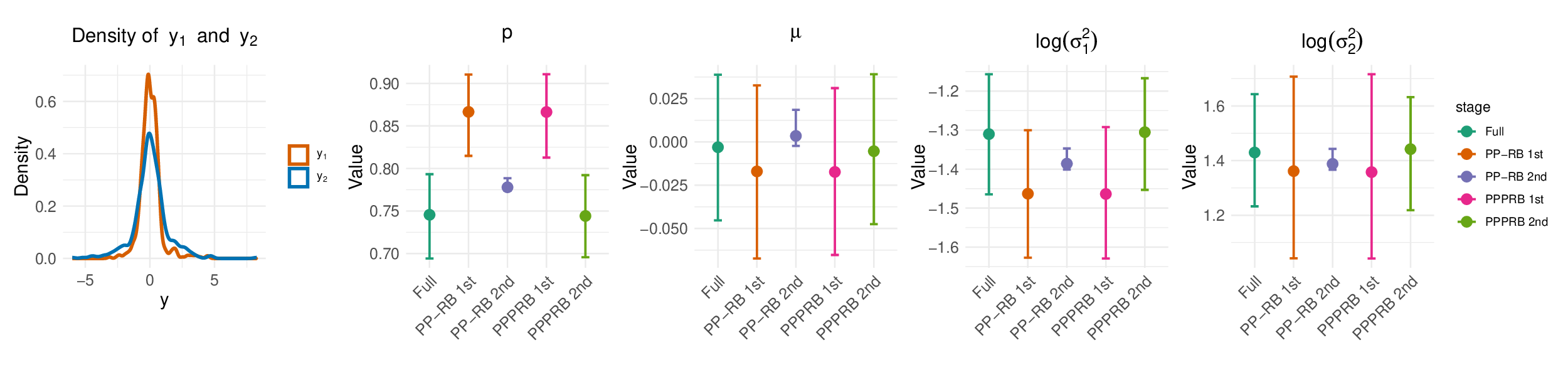}
    \caption{Numerical study results from Section \ref{simul_2}. Left: Density plot of simulated data with two partitions, $\bm{y}_1$ and $\bm{y}_2$, highlighted in orange and blue, respectively. Second to fifth: posterior summaries across algorithms for $p$, $\mu$, $\log \sigma_1^2$, and $\log \sigma_2^2$, respectively.}
    \label{fig:simul1_b}

\end{subfigure}
\vspace{0.5em} 
\begin{subfigure}{1\textwidth}
    \centering
    \includegraphics[width=\linewidth, trim=0 20 0 10, clip]{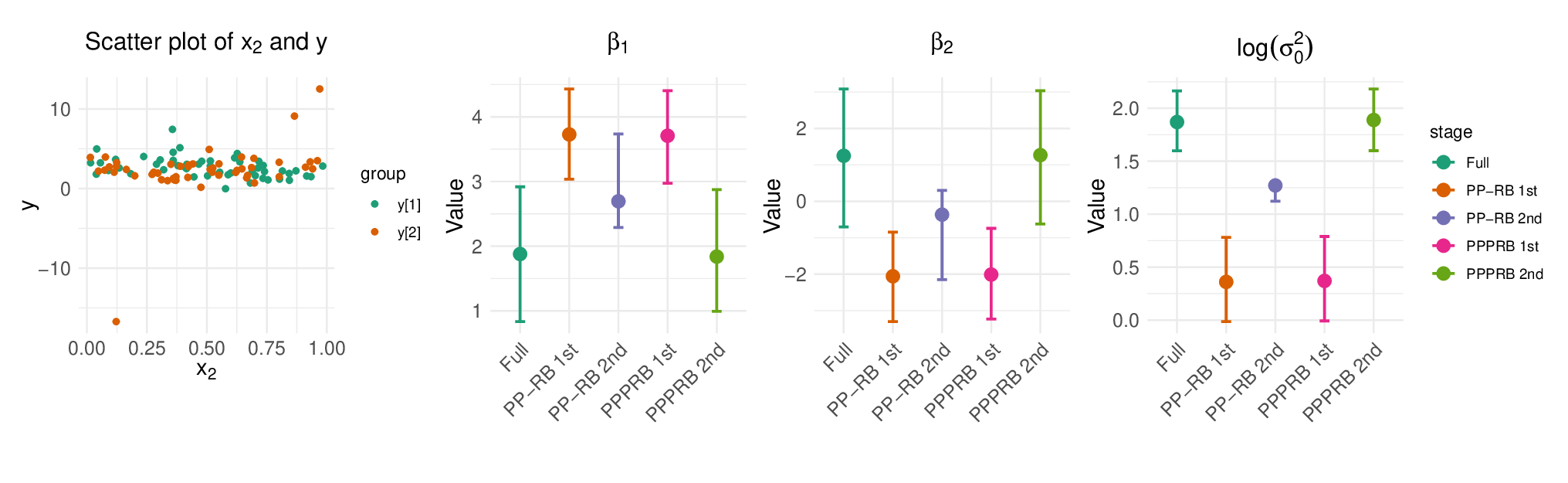}
    \caption{Numerical study results from Section \ref{simul_1}. Left: simulated data with two partitions, $\bm{y}_1$ and $\bm{y}_2$, highlighted in orange and green, respectively. Second to fourth: posterior summaries across algorithms for $\beta_1$, $\beta_2$, and $\log \sigma_0^2$, respectively.}
    \label{fig:simul1_a}
\end{subfigure}
\caption{Comparison of algorithms for the numerical studies in Sections \ref{simul_2} and \ref{simul_1}. MH, PP-RB (1st, 2nd), and PPP-RB (1st, 2nd) denote the MH algorithm, and the first- and second-stages of PP-RB and PPP-RB, respectively.}
\end{figure}
\subsection{Model misspecification}
\label{simul_1}
For the second, we consider the case where the model is misspecified. We simulated $N=100$ observations as
\[
y_i\sim p\cdot \text{Gau}(\bm{x}_i' \bm{\beta},\sigma^2_1) + (1-p) \cdot \text{Gau}(\bm{x}_i' \bm{\beta},\sigma^2_2), \quad i=1,...,N,
\]
where $p = 0.9$, $\sigma_1^2 = 1$, $\sigma_2^2 = 100$, and $\bm{\beta}=(3,-1)^{'}$. The vector $\bm{x}_i$ consisted of an intercept term and a random sample drawn from a uniform distribution. After generating data, we fit the linear model: $y_i =\bm{x}_i' \bm{\beta}+ \epsilon_i^*$, $\epsilon_i^*\sim \text{Gau}(0,\sigma^2_0)$ using a MH algorithm, PP-RB, and PPP-RB. For PP-RB and PPP-RB, we randomly partitioned the data into two subsets and considered two-stages. As shown in the left panel of Figure \ref{fig:simul1_a}, the second partition $\bm{y}_2$ included outliers. This can lead to noticeably different estimates of $\bm{\beta}$ across the two subsets. 

For all algorithms, we used vague priors $\bm{\beta} \sim \text{Gau}(\bm{0}, 1000\bm{I})$ and $\sigma_0^2 \sim \text{IG}(0.1, 0.1)$. We drew 50,000 MCMC samples, discarding the first 5,000 as burn-in. The resulting acceptance ratios ranged between 0.2 and 0.4. 
The setting for PPP-RB was the same as that in Section \ref{simul_2}.

From Figure \ref{fig:simul1_a}, it can be seen that the posterior distributions at the first stage for both PP-RB and PPP-RB deviated significantly from those obtained using the MH algorithm. Nonetheless, at the second stage, the posterior distribution from PPP-RB coincided with that of the MH algorithm due to the swapping of posterior samples between cold and hot chains. In contrast, the PP-RB posterior distribution remained significantly different because it was not able to explore the parameter space thoroughly.

\section{Data Analysis}
\label{real_data_analysis}
Because PP-RB and PPP-RB are designed to target the true posterior distribution in a more computationally efficient manner, it is important to first assess how closely their resulting posterior distributions agree with that of the MH algorithm, and then evaluate their efficiency in terms of effective sample size per unit time. We assessed these using two applications in this Section. The first involves a Hawkes point process model applied to earthquake frequency data, while the second employs a Gaussian Process (GP) model for sea surface salinity (SSS) data.
\subsection{1989 Loma Prieta Earthquake}
\label{hawkes}
In 1989, one of the most significant earthquakes in Northern California, the Loma Prieta earthquake, also known as the ``World Series earthquake,” occurred in the San Francisco Bay Area. This earthquake had a surface wave magnitude of 7.1 and led to 62 confirmed fatalities, with property damage and recovery costs estimated at \$6 billion \citep{us1990loma}. As an earthquake occurs, it often triggers subsequent earthquakes (aftershocks) in the surrounding region \citep{das2003spatial}. This seismic activity is often modeled using a Hawkes process \citep{hawkes1971spectra}, where each earthquake increases the conditional intensity of future earthquakes for some period of time, capturing the self-exciting nature of seismic events \citep[see, e.g.,][for applications of the Hawkes process to earthquake data]{kwon2023flexible, davis2024fractional,iwata2025mixture}.

A Hawkes process is a point process whose conditional intensity at time $t$ depends on the history up to time $t$. Formally, this can be represented by the conditional intensity function as: 
\[
\lambda(t) = \mu(t) + \sum_{t_i:t_i <t} \psi (t-t_i),
\]
where $\mu(t)\geq 0$ is the background intensity and $\psi(\cdot)$ is called a triggering kernel function, which controls how the past events affect the conditional intensity. When $\psi(\cdot)>0$, the process is called ``self-exciting,'' whereas when $\psi(\cdot)<0$, it is called ``self-regulating,'' indicating that past events respectively increase or decrease the conditional intensity, and hence the likelihood of future events. In this application, we express the Hawkes intensity as 
\[
\lambda(t)=\mu + \sum_{t_i : t_i<t} \big( \alpha \exp (-\beta (t-t_i)) \big), \quad \mu ,\alpha,\beta>0,
\]
where $\mu$, $\alpha$, and $\beta$ can be understood as the background intensity, the strength of self-excitation, and decay rate, respectively. In what follows, we use the re-parameterization $\eta=\alpha/\beta$, known as the branching ratio, which corresponds to the expected number of offspring events triggered by a single event. In our application, we assume $\eta<1$ to ensure the process remains stable over time.

We obtained earthquake data with magnitudes greater than 2.5 in 1989-1991 from \citet{USGS2026}. We focused on the region within 500 km of the main shock of the 1989 Loma Prieta earthquake (longitude -121.88, latitude 37.04) and further focused on earthquakes that occurred within California (longitude $\in$ [-125, -114], latitude $\in$ [32, 42]) between 2 Jan 1989 and 18 Oct 1990, UTC, resulting in a total of $N=2,250$ events during this period. The top panel in Figure \ref{fig_hawkes} shows a time series plot where the occurrence time of the 1989 Loma Prieta Earthquake is indicated by a blue dashed vertical line.
\begin{figure}[t]
\centering
\includegraphics[width=1\linewidth]{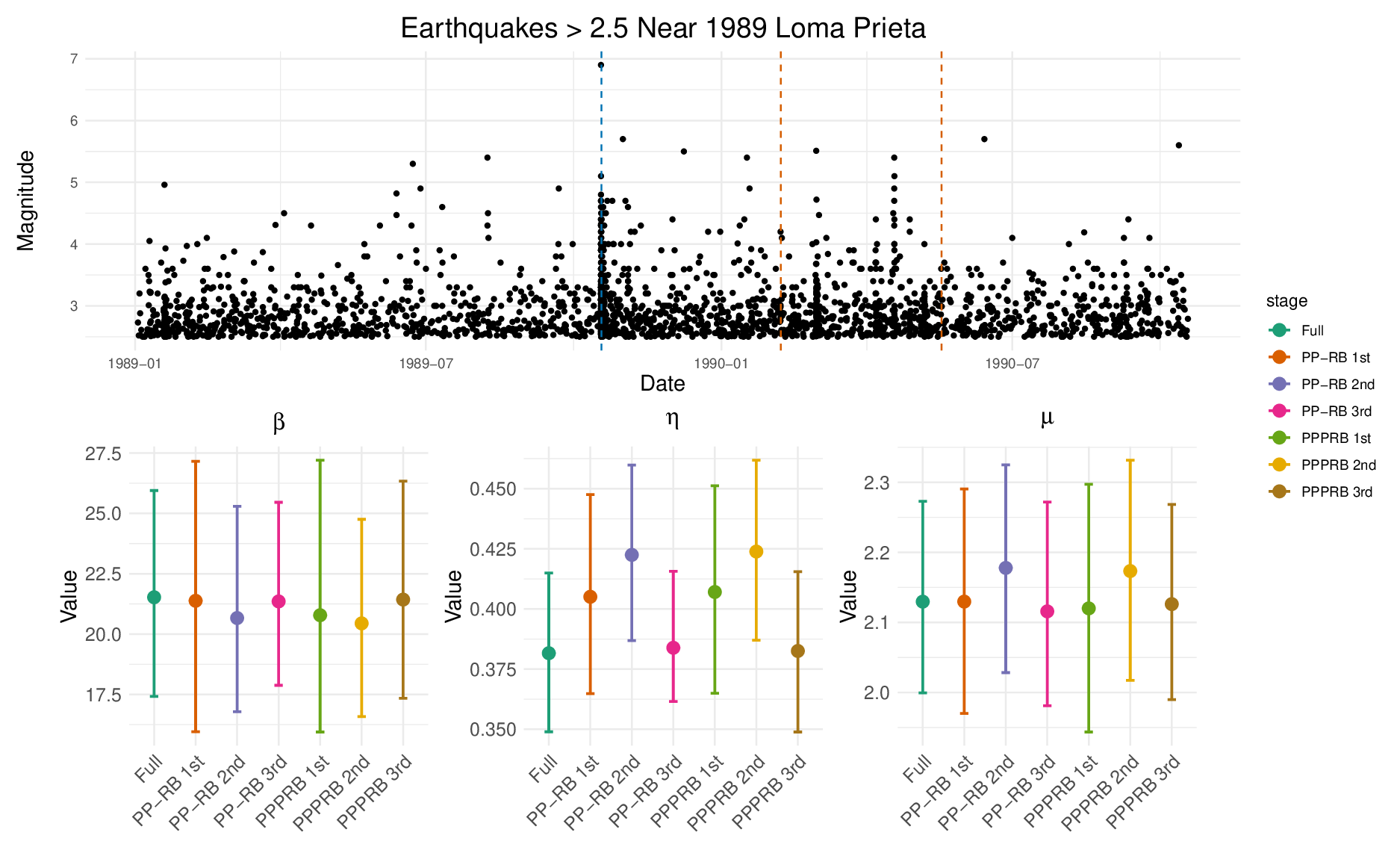}
\caption{Top: Time series plot of earthquake data. Orange dashed vertical lines indicate the time points at which the data are partitioned. Bottom left, bottom middle and bottom right: Posterior summaries comparison for standard MCMC (full data), PP-RB, and PPP-RB, respectively.}
\label{fig_hawkes}
\end{figure}

We assumed that the data arrived in an online setting, and our goal is to recursively update the posterior distribution so that the model does not need to be refit to the full dataset. Specifically, We assumed that data collected prior to 6 February 1990 (UTC) are available initially, while data from 6 February 1990 (UTC) to 17 May 1990 (UTC), and data after 17 May 1990 (UTC), become available sequentially in batches. Therefore, we considered three stages for PPP-RB and PP-RB, and MH algorithm, for comparison.

For all algorithms, we assumed the prior distributions $\beta \sim \text{Gamma}(2,0.5)$, $\mu\sim \text{Gamma}(1,1)$, and $\eta \sim \text{Beta}(2,2)$ and we drew 30,000 MCMC samples discarding the first 5,000 as burn-in. For PPP-RB, we considered $L=10$ chains (one cold and 9 hot chains). The temperature $\{\tau_{\ell}\}_{\ell=1}^{10}$ is defined by $T_{1}=1$ and $\tau_{\ell}=\exp(s_{\ell})$ for $\ell=2,...,10$, where $s_{\ell}$ are evenly spaced points over the interval $[0,2]$. The acceptance ratios ranged from 0.2 and 0.4.

To measure the agreement between posterior samples, we used the $\widehat{R}$ statistic \citep{gelman1992inference}. In other words, we used the potential scale reduction factor diagnostic to assess whether PP-RB and PPP-RB produce posterior samples comparable to those obtained from the MH algorithm. We also used the convergence diagnostic measure $\widehat{R}^*$ proposed in \citet{vehtari2021rank}. This measure modifies the $\widehat{R}$ statistic by applying rank normalization and folding to the posterior draws \citep{vehtari2021rank}. In terms of these two measures, both PPP-RB and PP-RB show good agreement with MH algorithm, with PPP-RB performing better. Specifically, for $\widehat{R}$, PPP-RB yields $\widehat{R}(\mu)=1.01,$ $\widehat{R}(\eta)=\widehat{R}(\beta)=1$, where $\widehat{R}(\mu)$, $\widehat{R}(\eta)$, and $\widehat{R}(\beta)$ denote the $\widehat{R}$ values for each parameter, whereas PP-RB yields $\widehat{R}(\mu)=1.07$, $\widehat{R}(\eta)=1.01$, and $\widehat{R}(\beta)=1.04$. PPP-RB is also slightly superior to PP-RB in terms of $\widehat{R}^*$, with $\widehat{R}^*(\mu)=\widehat{R}^*(\eta)=\widehat{R}^*(\beta)=1$ for PPP-RB, whereas $\widehat{R}^*(\mu)=\widehat{R}^*(\eta)=\widehat{R}^*(\beta)=1.01$ for PP-RB. Posterior summaries for each algorithm are presented in Figure \ref{fig_hawkes}.

\begin{table}[t]
\caption{
ESS and ET comparison across algorithms.
For PP-RB and PPP-RB, ET denotes total computation time over all stages.
}
\centering
\small
\setlength{\tabcolsep}{4pt}

\begin{tabular}{lccc@{\hspace{4mm}}ccc@{\hspace{4mm}}ccc}
\toprule

\multirow{2}{*}{Algorithm}
& \multicolumn{3}{c}{$\mu$}
& \multicolumn{3}{c}{$\eta$}
& \multicolumn{3}{c}{$\beta$} \\

\cmidrule(lr){2-4}
\cmidrule(lr){5-7}
\cmidrule(lr){8-10}

& ESS & ET & ESS/ET
& ESS & ET & ESS/ET
& ESS & ET & ESS/ET \\

\midrule

Full
& 708.1 & 806.2 & 0.9
& 1229.5 & 806.2 & 1.5
& 246.2 & 806.2 & 0.3 \\

PP-RB
& 211.9 & 344.3 & 0.6
& 228.1 & 344.3 & 0.7
& 237.5 & 344.3 & 0.7 \\

PPP-RB
& 2901.3 & 397.5 & 7.3
& 1963.1 & 397.5 & 4.9
& 1839.4 & 397.5 & 4.6 \\

\bottomrule
\end{tabular}
\label{table1}
\end{table}

We also investigated the computational cost for each algorithm. In particular, the measure we use for comparison is $\frac{\text{Effective Sample Size (ESS)}}{\text{Elapsed time (ET)}}$, which represents the number of effective samples obtained per unit time and thus serves as a measure of computational efficiency. From Table \ref{table1}, it can be seen that PPP-RB yields the highest ESS among all algorithms for all parameters, while PP-RB outperforms the others in terms of ET. In terms of $\text{ESS}/\text{ET}$, PPP-RB is the most efficient, with substantially larger $\text{ESS}/\text{ET}$ values compared to both MH and PP-RB. 

Algorithms were further compared for the case where data were partitioned into two subsets based on the occurrence time of the 1989 Loma Prieta Earthquake in the supplemental material. In this setting, we observed that the posterior distribution discrepancy at each stage is more severe, with PP-RB showing less accurate estimates, whereas PPP-RB still provides posterior distribution that aligns well with the posterior summary from MH algorithm, demonstrating its robustness.

\subsection{Sea Surface Salinity}
Sea surface salinity (SSS) quantifies the salinity at the ocean's surface and plays a crucial role in atmosphere-ocean interactions and vertical ocean circulation \citep{durack2016keeping,dinnat2019remote}. Its large-scale variations are primarily determined by evaporation, precipitation, and oceanic circulation \citep{terray2012near}. SSS increases or decreases as freshwater is added by precipitation or removed by evaporation, thus serving as an indicator of the marine hydrological cycle \citep{gordon2008sea}. SSS also provides valuable insights into the global water cycle and ocean dynamics; therefore, it is important to monitor SSS \citep{barale2010oceanography}.

\begin{figure}[t]
\centering
\includegraphics[width=1\linewidth]{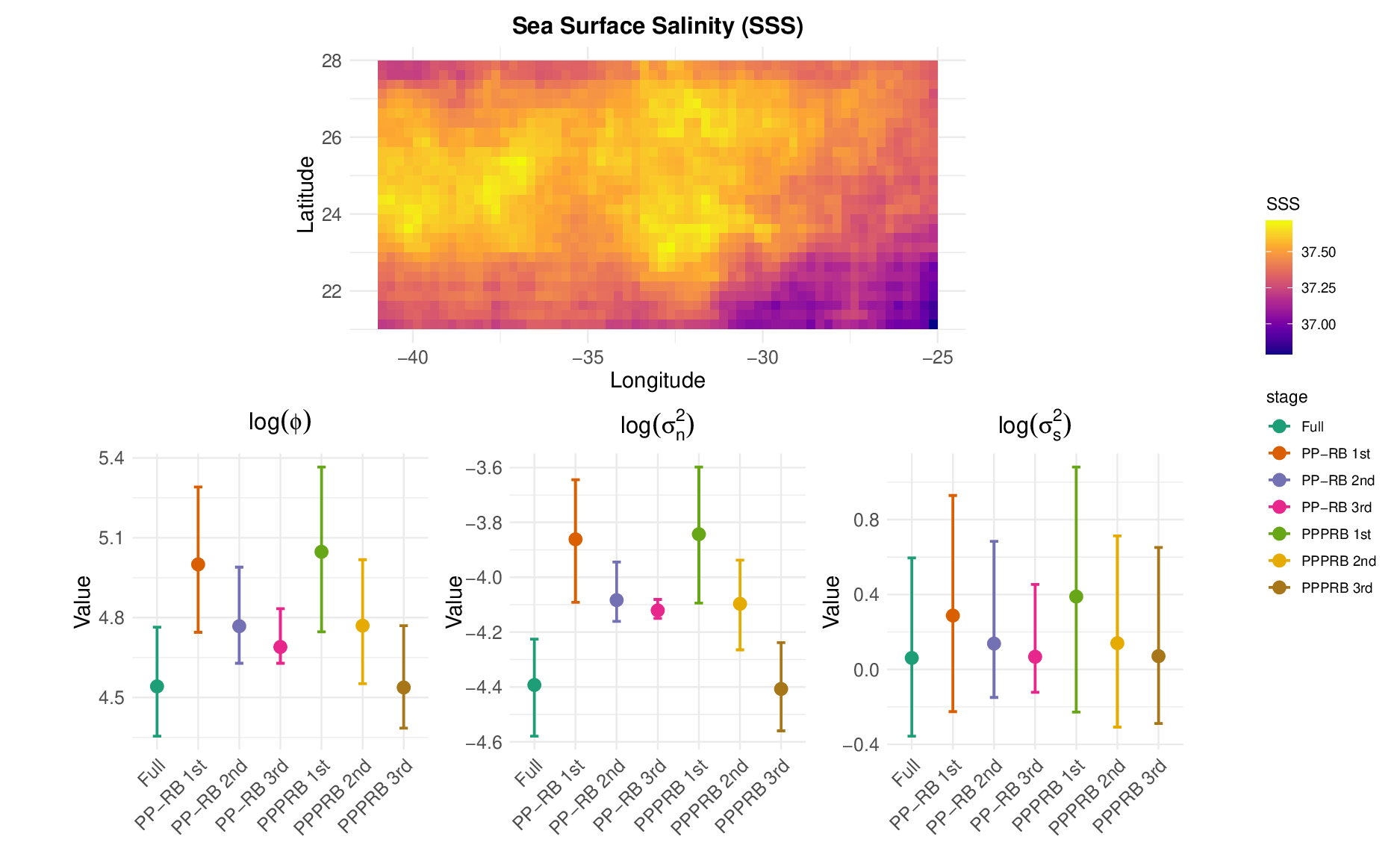}
\caption{Top: SSS data for 15 July 2012 on the North Atlantic region. Bottom left, Bottom middle, and Bottom right: Posterior summaries comparison for standard MCMC (full data), PP-RB, and PPP-RB, respectively.}
\label{fig_gp}
\end{figure}

\begin{table}[t]
\caption{
ESS and ET comparison across algorithms.
For PPP-RB, ET denotes total computation time over all stages.
}
\centering
\small
\setlength{\tabcolsep}{4pt}

\begin{tabular}{lccc@{\hspace{4mm}}ccc@{\hspace{4mm}}ccc}
\toprule

\multirow{2}{*}{Algorithm}
& \multicolumn{3}{c}{log($\sigma^2_s$)}
& \multicolumn{3}{c}{log($\sigma^2_n$)}
& \multicolumn{3}{c}{log($\phi$)} \\

\cmidrule(lr){2-4}
\cmidrule(lr){5-7}
\cmidrule(lr){8-10}

& ESS & ET & ESS/ET
& ESS & ET & ESS/ET
& ESS & ET & ESS/ET \\

\midrule

MH
& 48.3   & 2685.3 & 0.02
& 588.1  & 2685.3 & 0.22
& 67.7   & 2685.3 & 0.03 \\

PPP-RB
& 3423.3 & 828.7 & 4.13
& 3220.4 & 828.7 & 3.89
& 2952.9 & 828.7 & 3.56 \\

\bottomrule
\end{tabular}
\label{table2}
\end{table}
We obtained SSS data from the ESA Sea Surface Salinity Climate Change Initiative \citep{Boutin2024} for 15 July 2012 and focused on the North Atlantic region, defined by longitudes $[-41,-25]$ and latitudes $[21,28]$, resulting in a total of $N = 1,792$ observations (top left panel in Figure \ref{fig_gp}). For PP-RB and PPP-RB, we randomly partitioned the data into three subsets, thus implementing a three-stage algorithm.

For the SSS data, we employed a Gaussian process (GP) model with a Mat\'ern \citep{matern2013spatial} covariance function, where the smoothness parameter is fixed at $3/2$, such that the covariance matrix is
\[
\bm{\Sigma} = \sigma^2 \left( (1 - \tau^2)\,\bm{R}(\phi) + \tau^2 \bm{I} \right),
\]
where the entries of the correlation matrix $\bm{R}(\phi)$ are given by $R_{ij}(\phi) = \left( 1 + \frac{d_{ij}}{\phi} \right)\exp\left( -\frac{d_{ij}}{\phi} \right)$, and $d_{ij}$ denotes the distance between $\bm{s}_i$ and $\bm{s}_j$. We reparametrize the covariance matrix above as
\[
\bm{\Sigma}=\sigma^2_s \bm{R}(\phi) + \sigma^2_n \bm{I},
\]
where $\sigma^2_s=\sigma^2(1-\tau^2)$ and $\sigma^2_n=\sigma^2 \tau^2$ are spatial variance and nugget term. With priors $\log \sigma^2_s \sim \text{Gau}(\log(0.6),1)$, $\log \sigma^2_n \sim \text{Gau}(\log(0.05),1)$, and $\log \phi \sim \text{Gau}(\log(d_{\text{med}}/5),1)$, where $d_{\text{med}}$ denotes the median of all pairwise distances between locations, we drew 30,000 MCMC samples discarding the first 5,000 as burn-in for each algorithm. For PPP-RB, we obtained $L=10$ chains in parallel with temperatures set as $T_1 = 1$ and $T_\ell = \exp(s_\ell)$ for $\ell = 2,\dots,10$, where $s_\ell$ are evenly spaced over the interval $[0,3]$. The acceptance ratios range from 0.2 and 0.4.

The posterior distribution from PPP-RB shows strong agreement with those from the MH algorithm in terms of $\widehat{R}$ and $\widehat{R}^*$. Specifically, $\widehat{R}(\log(\phi)) = \widehat{R}(\log(\sigma_s^2)) = 1$, $\widehat{R}(\log(\sigma_n^2)) = 1.01$, and $\widehat{R}^*(\log(\phi)) = \widehat{R}^*(\log(\sigma_s^2)) = 1$, $\widehat{R}^*(\log(\sigma_n^2)) = 1.01$. In contrast, PP-RB yields posterior distributions that differ substantially, with $\widehat{R}(\log(\phi)) = 2.49$, $\widehat{R}(\log(\sigma_s^2)) = 1.07$, $\widehat{R}(\log(\sigma_n^2)) = 5.53$, and $\widehat{R}^*(\log(\phi)) = 1.36$, $\widehat{R}^*(\log(\sigma_s^2)) = 1.13$, $\widehat{R}^*(\log(\sigma_n^2)) = 1.90$. The Top right and bottom panels in Figure \ref{fig_gp} show the posterior summaries for each algorithm. Notably, the posterior summaries from MH algorithm differ markedly from those of the first stage of PP-RB and PPP-RB. As a result of this discrepancy, PP-RB fails to adequately capture the posterior summary obtained from MH algorithm. In contrast, PPP-RB was able to capture the posterior summaries accurately, demonstrating more robust inference than PP-RB.

PP-RB failed to yield accurate posterior inferences, thus we compare only the MH algorithm and PPP-RB in terms of computational efficiency in Table \ref{table2}. PPP-RB was almost 3.2 times more efficient than the MH algorithm in terms of ET. Furthermore, due to the swapping of posterior samples between the cold and hot chains, PPP-RB achieves a significantly higher ESS than the MH algorithm.

\section{Discussion}
\label{discussion}
Bayesian inference is appealing because it allows us to use a coherent statistical framework to learn about data generating processes. In particular, conditional representations inherent to Bayesian hierarchical models \citep{berliner1996hierarchical} provide a highly flexible framework for modeling dependent data. Although recent advances in computational resources have facilitated Bayesian inference, standard algorithms such as MCMC remain computationally prohibitive for large datasets. This has motivated growing interest in scalable algorithms for fitting Bayesian models, including approximation methods such as variational inference \citep{jordan1999introduction}, integrated nested Laplace approximation \citep{rue2009approximate}, and subsampling-based MCMC \citep{bradley2021approach}.

The Prior Proposal-Recursive Bayesian inference \citep{hooten2021making} involves a recursive algorithm that targets the exact posterior distribution. By leveraging parallel computing resources and updating the posterior distribution sequentially across stages, this scalable approach makes Bayesian inference more accessible to a wide range of applications. Despite its strengths, the effective sample size of PP-RB can decrease quickly. More critically, PP-RB can result in estimates that deviate from the true posterior when posterior distributions across stages differ substantially from one another.

To address these limitations and enable more robust inference, we proposed PPP-RB, which integrates PP-RB with MCMCMC \citep{geyer1991mcmc}. We showed that PPP-RB targets the true posterior distribution and discussed practical considerations for its implementation. In numerical studies, we compared PP-RB and PPP-RB, highlighting that PPP-RB effectively addressed the limitations of PP-RB. We also demonstrated the advantages of PPP-RB over PP-RB and the standard MH algorithm in terms of effective sample size and computational time through data analysis.

There are multiple ways to extend PPP-RB. First, developing a theoretical framework for tuning the algorithm would be desirable. Although we provided guidelines to ease the implementation, theoretical guidelines could further improve performance and ensure optimal tuning. Second, the optimal partitioning scheme could improve the efficiency of PPP-RB. While the partitioning scheme is fixed under certain settings (e.g., the online setting), data are typically randomly partitioned. Although we showed that hot chains in PPP-RB can effectively explore the parameter space, resulting in the correct inference, the impact of alternative partitioning strategies on efficiency and mixing remains an open question. 


\bibliographystyle{apalike}  
\bibliography{references}  

\clearpage
\appendix

\section{Supplementary Material}
\subsection{Details of the results presented in Section 2.2.}
\subsubsection{Temperature optimization criterion}\label{temp_opti}

For a tempered proposal of the form $[\bs{\theta} \mid \bs{y}_{1}]_{\gamma} \propto [\bs{y}_{1} \mid \bs{\theta}]^{\gamma} [\bs{\theta}]$, where $\gamma = 1/\tau$, selecting $\gamma$ is crucial to address the limitation of PP-RB (i.e., posterior shift). One criterion is to minimize the discrepancy between $[\bm{\theta} \mid \bm{y}_{1:J}]$ and the proposal $[\bm{\theta} \mid \bm{y}_{1}]_{\gamma}$. Among the various choices of discrepancy measures (e.g., total variation distance, Kullback-Leibler divergence, or the $L_2$ norm), we focus on the $\chi^2-$divergence, defined as 
\begin{align}
D_{\chi^{2}}\Big( [\bs{\theta} \mid \bs{y}_{1:J}] \Big\| [\bs{\theta} \mid \bs{y}_{1}]_{\gamma} \Big)=\int \frac{\{[\bs{\theta} \mid \bs{y}_{1:J}]\}^{2}}{[\bs{\theta} \mid \bs{y}_{1}]_{\gamma}} d\bm{\theta} - 1.
\end{align}
Then, the optimal $\gamma^*$ is such that 
\begin{align}
\label{optimal_temp}
\gamma^{*} &= \arg \min_{\gamma \geq 0} D_{\chi^{2}}\Big( [\bs{\theta} \mid \bs{y}_{1:J}] \Big\| [\bs{\theta} \mid \bs{y}_{1}]_{\gamma} \Big).
\end{align}
Note that Equation \eqref{optimal_temp} is also related to finding $\gamma^*$ such that it maximizes a rough approximation of the effective sample size \citep{liu1995blind,martino2017effective} defined as Equation \eqref{ess} in Section \ref{snis_pprb} because \begin{align*}
\arg \max_{\gamma \geq 0} \widehat{\text{ESS}}_{\gamma}
        =  \arg \max_{\gamma \geq 0} M\Big\{ 1 + \text{Var}\Big( \widetilde{w}_{\gamma}(\bs{\theta}_{1}) \Big) \Big\}^{-1}, \quad 
        \widetilde{w}_{\gamma}(\bs{\theta}) = \frac{[\bs{\theta} \mid \bs{y}_{1:J}]}{[\bs{\theta} \mid \bs{y}_{1}]_{\gamma}}
\end{align*}
and $\text{Var}\Big( \widetilde{w}_{\gamma}(\bs{\theta}_{1}) \Big) = D_{\chi^{2}}\Big( [\bs{\theta} | \bs{y}_{1:J}] \Big\| [\bs{\theta} | \bs{y}_{1}]_{\gamma} \Big)$. It is worth noting that Equation \eqref{optimal_temp} requires   $[\bm{\theta} \mid \bm{y}_{1:J}]$ , which is generally unavailable except in certain special cases.

\subsubsection{Asymptotic equivalence between PP-RB and SNIS}
\label{snis_pprb}
We show the equivalence between PP-RB and Self-normalized Importance sampling (SNIS). Suppose that we are interested in calculating expectations of the form
\begin{align*}
I(f) = \int f(\bm{\theta})  [\bm{\theta} \mid \bm{y}_{1:n}] \ d\bm{\theta}.
\end{align*}
We often approximate the expression above using a Monte Carlo estimator of the form 
\begin{align*}
\widetilde{I}(f) = \sum_{m=1}^{M} w_{m}(\bm{\theta}_{1:M}) f(\bm{\theta}_{m}),
\end{align*}
where $\bm{\theta}_{1}, \dots, \bm{\theta}_{M}$ are samples from a proposal distribution $[\bm{\theta}]_{\text{prop}}$ and $ w_{m}(\bm{\theta}_{1:M})$ is the importance weight. The usual Monte Carlo estimator is given by
\begin{align*}
\widehat{I}(f) = \frac{1}{M} \sum_{m=1}^{M} f(\bm{\theta}_{m}), 
\end{align*}
where $\bm{\theta}_{1}, \dots, \bs{\theta}_{M} \overset{\text{i.i.d.}}{\sim} [\bm{\theta} | \bm{y}_{1:n}]$. $\widetilde{I}(f)$ is often compared to $\widehat{I}(f)$ by the effective sample size, defined as
\begin{align}
\label{ess}
\text{ESS}\left( \tilde{I}(f) \right) = M\frac{\text{Var}\left( \hat{I}(f) \right)}{\text{Var}\left( \tilde{I}(f) \right)},
\end{align}
which can be interpreted as the number of i.i.d. samples from $[\bs{\theta} | \bs{y}_{1:n}]$ that would be needed to attain the same Monte Carlo error as $\tilde{I}(f)$. 

In PP-RB, posterior samples are drawn from $[\bm{\theta} \mid \bm{y}_1 ]$ at the first stage. In subsequent stages, proposals $\bm{\theta}^*$ are drawn with replacement from the posterior samples of the previous stage. Specifically, at the $m$th iteration of the $j$th stage, a proposal is accepted according to the MH acceptance ratio $\min(1,r)$ where $r$ is given by
\begin{align*}
r= \frac{[\bm{\theta}^* \mid \bs{y}_{1:j}] [\bm{\theta}^{(m-1)} \mid \bm{y}_{1:(j-1)}]}{[\bm{\theta}^{(m-1)} \mid \bs{y}_{1:j}] [\bs{\theta}^* \mid \bs{y}_{1:(j-1)}]},
\end{align*}
and $\bm{\theta}^{(m-1)}$ denote the posterior sample at the $(m-1)$th iteration. Note that $[\bm{\theta}^{(m-1)} \mid \bm{y}_{1:(j-1)}]$ is discrete uniform distribution. If we instead propose at the 
$j$th stage from
\begin{align*}
\boldsymbol{\theta}^* \sim \text{Categorical}\left(\{w_j(\boldsymbol{\theta}^{(i)})\}_{i=1}^{M},\ \{\boldsymbol{\theta}^{(i)}\}_{i=1}^{M}\right),
\end{align*}
where $\{\bm{\theta}^{(i)}\}_{i=1}^{M}$
denotes the $M$ posterior samples from the $(j-1)$th stage and
\begin{align*}
w_j(\boldsymbol{\theta}) = \frac{\tilde{w}_j(\boldsymbol{\theta})}{\sum_{i=1}^{M}\tilde{w}_j(\boldsymbol{\theta}^{(i)})}, \qquad \tilde{w}_j(\boldsymbol{\theta}) = \frac{[\boldsymbol{\theta} \mid \boldsymbol{y}_{1:j}]}{[\boldsymbol{\theta} \mid \boldsymbol{y}_{1:(j-1)}]}.
\end{align*}
With this, the MH acceptance ratio becomes $\min(1,r)$, where $r$ is
\begin{align*}
r&= \frac{[\bm{\theta}^* \mid \bs{y}_{1:j}] [\bm{\theta}^{(m-1)} \mid \bm{y}_{1:(j-1)}]}{[\bs{\theta}^{(m-1)} \mid \bm{y}_{1:j}] [\bm{\theta}^* \mid \bm{y}_{1:(j-1)}]} \cdot \frac{w_{j}(\bm{\theta}^{(m-1)})}{w_{j}(\bm{\theta}^*)}\\
&= \frac{[\bm{\theta}^* \mid \bs{y}_{1:j}] [\bm{\theta}^{(m-1)} \mid \bm{y}_{1:(j-1)}]}{[\bs{\theta}^{(m-1)} \mid \bm{y}_{1:j}] [\bm{\theta}^* \mid \bm{y}_{1:(j-1)}]} \cdot \frac{\tilde{w}_{j}(\bs{\theta}^{(m-1)})}{\tilde{w}_{j}(\bm{\theta}^*)} \\
&= \frac{[\bm{\theta}^* \mid \bs{y}_{1:j}] [\bm{\theta}^{(m-1)} \mid \bm{y}_{1:(j-1)}]}{[\bs{\theta}^{(m-1)} \mid \bm{y}_{1:j}] [\bm{\theta}^* \mid \bm{y}_{1:(j-1)}]} \cdot \frac{ \frac{[\bs{\theta}^{(m-1)} \mid \bs{y}_{1:j}]}{[\bs{\theta}^{(m-1)} \mid \bs{y}_{1:(j-1)}]} }{ \frac{[\bm{\theta}^* \mid \bs{y}_{1:j}]}{[ \bm{\theta}^* | \bs{y}_{1:(j-1)}]} } \\
&= 1.
\end{align*}
Therefore, $[\bm{\theta}]^{*}_{j} = \text{Categorical}\left(  \{ w_{j}(\bm{\theta}^{(i)}) \}_{i=1}^M, \{ \bm{\theta}^{(i)} \}_{i=1}^{M} \right)$ is the stationary distribution of the chain. This implies that, when considering the idealized scenario in which the samples from each stage are thinned enough to make the Markovian dependence negligible, PP-RB can be seen as a sequence of sampling importance resampling (SIR) steps, where the $j$th step adjusts the importance sampling weights to account for the information in $\bm{y}_j$. In particular, the samples generated from a two-stage PP-RB are equivalent to those from a SIR procedure that generates proposals $\bm{\theta}_1,...,\bm{\theta}_M\overset{\text{i.i.d.}}{\sim} [\bs{\theta} \mid \bm{y}_{1}] $ and resamples them with probability proportional to the self-normalized importance sampling (SNIS) weights $\{ w_2 (\bm{\theta}^{(m)}) \}_{m=1}^{M}$. For this reason, we can approximately assess the expected performance of the Monte Carlo estimators obtained via PP-RB by comparing them to the analogous SNIS estimators.

\subsubsection{Illustration}
\textbf{Optimal temperature based on $\chi^2-$ divergence:}
We illustrate, through a simple example, how the criterion in Equation \eqref{optimal_temp} can be applied in practice to select the optimal value of $\gamma$. Consider a set of observations divided into $J$ partitions $\bs{y}_{1}, \dots, \bs{y}_{J}$, of size $n_{1}, \dots, n_{J}$. We assume the Bayesian model
\begin{align}
\label{bayesian_model}
\bm{y}_{ji} \mid \bm{\theta} &\sim \text{Gau}( \bm{\theta}, \bm{\Sigma} ), \\
\bm{\theta} &\sim \text{Gau}( \bm{\theta}_{0}, \bm{\Sigma}_{0}) \nonumber ,
\end{align}
where the observations $\bm{y}_{ji} \in \mathbb{R}^d$ are assumed to be independent given $\bm{\theta}$ for all $j \in \{1, \dots, J\}$ and $i \in \{1, \dots, n_{j}\}$, and $\bm{\Sigma}$ is assumed to be known. In this case, the posterior distribution can be analytically derived as below. 
\begin{align*}
[\bm{\theta} \mid \bm{y}_{1:J}] &= \text{Gau}\left( \bm{\theta}_{J}, \bm{\Sigma}_{J} \right), \\
[\bm{\theta} | \bm{y}_{1}]_{\gamma} &= \text{Gau}\left( \bm{\theta}_{1,\gamma}, \bm{\Sigma}_{1,\gamma} \right),
\end{align*}
where $\bm{\Sigma}_{J} = \bm{\Sigma}_{0} + n\bm{\Sigma}$, $\bm{\theta}_{J} = \bm{\Sigma}_{J}^{-1} (\bm{\Sigma}_{0} \bm{\theta}_{0} + n \bm{\Sigma} \bar{\bm{y}}_{J})$, $\bm{\Sigma}_{1,\gamma} = \bm{\Sigma}_{0} + \gamma n_{1} \bm{\Sigma}$, $\bm{\theta}_{1,\gamma} = \bm{\Sigma}_{1, \gamma}^{-1} ( \bm{\Sigma}_{0} \bm{\theta}_{0} + \gamma n_{1} \bm{\Sigma} \bar{\bm{y}}_{1} )$, $n = \sum_{j=1}^{J} n_{j}$, $\bar{\bm{y}}_{J} = \frac{1}{n} \sum_{j=1}^{J} \sum_{i=1}^{n_{j}} \bm{y}_{ji}$, and $\bar{\bm{y}}_{1} = \frac{1}{n_{1}} \sum_{i=1}^{n_{1}} \bm{y}_{1i}$. Then, assuming vague prior for $\bm{\Sigma}_0$, it can be shown that 
\begin{align*}
D_{\chi^{2}}\Big( [\bm{\theta} \mid \bm{y}_{1:J}] \Big\| [\bs{\theta} \mid \bm{y}_{1}]_{\gamma} \Big) 
        &= \int \frac{[\bm{\theta} \mid \bm{y}_{1:J}]^{2}}{[\bm{\theta} \mid \bs{y}_{1}]_{\gamma}} d\bs{\theta} - 1 \\
        & = \frac{|\bs{\Sigma}_{J}|}{|\tilde{\bs{\Sigma}}|^{\frac{1}{2}} |\bs{\Sigma}_{1}|^{\frac{1}{2}}} \exp\left\{ \tfrac{1}{2} (\bs{\theta}_{J} - \bs{\theta}_{1,\gamma})' \left( \bs{\Sigma}_{1,\gamma}^{-1} - \tfrac{1}{2}\bs{\Sigma}_{J}^{-1} \right)^{-1} (\bs{\theta}_{J} - \bs{\theta}_{1,\gamma}) \right\} - 1,
\end{align*}
where $\tilde{\bm{\Sigma}}=2 \bm{\Sigma}_J- \bm{\Sigma}_{1,\gamma}$. Then, the optimal $\gamma^*$ is
\begin{align}
\label{optimal_gamma}
\gamma^{*} = \frac{n}{n_1} \left\{ \frac{3}{2} + \frac{1}{d} \text{tr}(S) - \frac{1}{2}\sqrt{1 + 4\left\{ \left( \frac{1}{d}\text{tr}(S) + \frac{3}{2} \right)^{2} - \left( \frac{3}{2} \right)^{2} \right\} } \right\},
\end{align}
where $\bm{S} = n \bs{\Sigma} (\bar{\bs{y}}_{J} - \bar{\bs{y}}_{1})(\bar{\bs{y}}_{J} - \bar{\bs{y}}_{1} )'$.

\textbf{Optimal temperature based on ESS:} As discussed in Section \ref{temp_opti}, the criterion in Equation \eqref{optimal_temp} is closely related to maximizing the rough approximation of the ESS. We also demonstrate how to obtain the optimal temperature from the ESS using the same example considered for the $\chi^2$ divergence. However, it is worth noting that the ESS-based optimal value of $\gamma^*$ depends on the quantity of interest (e.g., the expectation), whereas the criterion in Equation \eqref{optimal_temp} is agnostic to the specific quantity of interest.

In the case where the quantity of interest is the expectation of the observations, the optimal value $\gamma^*$ can be obtained by maximizing the ESS, namely,
\begin{align*}
       \gamma^{*} = \arg \max_{\gamma \geq 0} \text{ESS}_{\gamma}(g)
\end{align*}
where $g(\bs{x}) = \bs{x}$. Using the analytic approximation of the ESS for SNIS estimators derived with the multivariate delta method (see Section \ref{sec:delta method}), we obtain an approximation of the multivariate ESS given by
\begin{align*}
        \text{ESS}_{\gamma}(g) 
        &\approx M \left( \frac{\left| \text{Var}\Big(g(\tilde{\bs{\theta}}_{1})\Big) \right|}{ \left| \mathbb{E}\Big( \tilde{w}^{2}_{\gamma}(\bs{\theta}_{1}) \left\{ g(\bs{\theta}_{1}) - I(g) \right\}\left\{ g(\bs{\theta}_{1}) - I(g) \right\}' \Big) \right| } \right)^{\frac{1}{d}} \\
        &= M \left| \bs{\Sigma}_{J}^{-1} \left\{ \mathbb{E}\Big( \tilde{w}_{\gamma}^{2}(\bs{\theta}_{1}) \left\{ g(\bs{\theta}_{1}) - I(g) \right\}\left\{ g(\bs{\theta}_{1}) - I(g) \right\}' \Big) \right\}^{-1} \right|^{\frac{1}{d}}.
\end{align*}
Note that
\begin{align*}
&\mathbb{E}\Big( \tilde{w}^{2}(\bs{\theta}_{1}) \left\{ g(\bs{\theta}_{1}) - I(g) \right\}\left\{ g(\bs{\theta}_{1}) - I(g) \right\}' \Big) \\
        &\quad= \int \tilde{w}_{\gamma}^{2}(\bs{\theta}) \left\{ g(\bs{\theta}) - I(g) \right\}\left\{ g(\bs{\theta}_{1}) - I(g) \right\}' [\bs{\theta} | \bs{y}_{1}]_{\gamma} \ d\bs{\theta} \\
        &\quad= \int  \left\{ g(\bs{\theta}) - I(g) \right\}\left\{ g(\bs{\theta}_{1}) - I(g) \right\}' \frac{[\bs{\theta} \mid \bs{y}_{1:J}]^{2}}{ [\bs{\theta} \mid \bs{y}_{1}]_{\gamma} } \ d\bs{\theta}
\end{align*}
Also,
\begin{align*}
         \frac{[\bs{\theta} | \bs{y}_{1:J}]^{2}}{ [\bs{\theta} | \bs{y}_{1}]_{\gamma} } 
         &= \frac{(2\pi)^{-\frac{2d}{2}} |\bs{\Sigma}_{J}| \exp\left\{ -\frac{1}{2} (\bs{\theta} - \bs{\theta}_{J})' 2\bs{\Sigma}_{J} (\bs{\theta} - \bs{\theta}_{J}) \right\} }{ (2\pi)^{-\frac{d}{2}} |\bs{\Sigma}_{1,\gamma}|^{\frac{1}{2}} \exp\left\{ -\frac{1}{2} (\bs{\theta} - \bs{\theta}_{1,\gamma})' \bs{\Sigma}_{1,\gamma} (\bs{\theta} - \bs{\theta}_{1.\gamma}) \right\} } \\
         &= \underbrace{ \frac{|\bs{\Sigma}_{J}|}{|\tilde{\bs{\Sigma}}|^{\frac{1}{2}} |\bs{\Sigma}_{1}|^{\frac{1}{2}}} \exp\left\{ \tfrac{1}{2} (\bs{\theta}_{J} - \bs{\theta}_{1,\gamma})' \left( \bs{\Sigma}_{1,\gamma}^{-1} - \tfrac{1}{2}\bs{\Sigma}_{J}^{-1} \right)^{-1} (\bs{\theta}_{J} - \bs{\theta}_{1,\gamma}) \right\} }_{:= c_{\gamma}(\bs{y}_{1:J})} \\
         &\quad \times (2\pi)^{-\frac{1}{2}} | \tilde{\bs{\Sigma}} |^{\frac{1}{2}} \exp\left\{ -\tfrac{1}{2} (\bs{\theta} - \tilde{\bs{\theta}})' \tilde{\bs{\Sigma}} (\bs{\theta} - \tilde{\bs{\theta}}) \right\},
\end{align*}
where $\tilde{\bs{\Sigma}} = 2\bs{\Sigma}_{J} - \bs{\Sigma}_{1,\gamma}$ and $\tilde{\bs{\theta}} = \tilde{\bs{\Sigma}}^{-1} (2\bs{\Sigma}_{J}\bs{\theta}_{J} - \bs{\Sigma}_{1,\gamma} \bs{\theta}_{1,\gamma})$. Therefore, assuming $\tilde{\bs{\Sigma}} \succeq 0$ (positive semidefinite), we have
\begin{align*}
       &\mathbb{E}\Big( \tilde{w}^{2}(\bs{\theta}_{1}) \left\{ g(\bs{\theta}_{1}) - I(g) \right\}\left\{ g(\bs{\theta}_{1}) - I(g) \right\}' \Big) = c_{\gamma}(\bs{y}_{1:J}) \left\{ \tilde{\bs{\Sigma}}^{-1} + \{ \tilde{\bs{\theta}} - I(g) \}\{ \tilde{\bs{\theta}} - I(g) \}' \right\},
\end{align*}
which cannot be computed unless $I(g)$ is known. Replacing it with an estimate (e.g., $\bs{\theta}_{J}$), yields the estimator
\begin{align}
\label{eq:ESS delta method}
\widehat{\text{ESS}}_{\gamma}(g) = M \{c_{\gamma}(\bs{y}_{1:J})\}^{-1} \left| \bs{\Sigma}_{n} \left\{ \tilde{\bs{\Sigma}}^{-1} + (\tilde{\bs{\theta}} - \bs{\theta}_{J})(\tilde{\bs{\theta}} - \bs{\theta}_{J})' \right\} \right|^{-\frac{1}{d}},
\end{align}
which can be computed, thus allowing us to optimize for $\gamma$. It is worth noting that, assuming the prior is sufficiently vague (i.e., as $\bs{\Sigma}_{0} \to \bs{0}$), this expression simplifies to
\begin{align*}
        \widehat{\text{ESS}}_{\gamma}(g) 
        &= M \frac{|\tilde{\bs{\Sigma}}|^{\frac{1}{2}} |\bs{\Sigma}_{1}|^{\frac{1}{2}}}{|\bs{\Sigma}_{J}|} \exp\left\{ -\tfrac{1}{2} (\bs{\theta}_{J} - \bs{\theta}_{1,\gamma})' \left( \bs{\Sigma}_{1,\gamma}^{-1} - \tfrac{1}{2}\bs{\Sigma}_{J}^{-1} \right)^{-1} (\bs{\theta}_{J} - \bs{\theta}_{1,\gamma}) \right\} \\
        &\quad\times \left| \bs{\Sigma}_{J} \left\{ \tilde{\bs{\Sigma}}^{-1} + (\tilde{\bs{\theta}} - \bs{\theta}_{J})(\tilde{\bs{\theta}} - \bs{\theta}_{J})' \right\} \right|^{-\frac{1}{d}} \\
        &\to M \frac{|(2n-\gamma n_{1})\bs{\Sigma}|^{\frac{1}{2}} |n_{1}\bs{\Sigma}|^{\frac{1}{2}}}{|n\bs{\Sigma}|} \exp\left\{ -\tfrac{1}{2} (\bar{\bs{y}}_{J} - \bar{\bs{y}}_{1})' \left( \left\{\tfrac{1}{\gamma n_{1}} - \tfrac{1}{2n} \right\}\bs{\Sigma}^{-1} \right)^{-1} (\bar{\bs{y}}_{J} - \bar{\bs{y}}_{1}) \right\} \\
        &\quad\times \left| n\bs{\Sigma} \left\{ \left\{ 2n\bs{\Sigma} - \gamma n_{1} \bs{\Sigma} \right\}^{-1} + \left(\frac{\gamma n_{1}}{2n - \gamma n_{1}}\right)^{2}(\bar{\bs{y}}_{J} - \bar{\bs{y}}_{1})(\bar{\bs{y}}_{J} - \bar{\bs{y}}_{1})' \right\} \right|^{-\frac{1}{d}} \\
        &= M \left\{ \frac{\gamma n_{1} (2n - \gamma n_{1})}{n^{2}} \right\}^{\frac{d}{2}} \exp\left\{ -\tfrac{\gamma n_{1} n}{2n-\gamma n_{1}} (\bar{\bs{y}}_{J} - \bar{\bs{y}}_{1})' \bs{\Sigma} (\bar{\bs{y}}_{J} - \bar{\bs{y}}_{1}) \right\} \\
        &\quad\times \left| \tfrac{n}{2n-\gamma n_{1}} I_{d} + \left( \tfrac{\gamma n_{1}}{2n-\gamma n_{1}} \right)^{2} n\bs{\Sigma}(\bar{\bs{y}}_{J} - \bar{\bs{y}}_{1})(\bar{\bs{y}}_{J} - \bar{\bs{y}}_{1} )' \right|^{-\frac{1}{d}} \\
        &= \left\{ \gamma \alpha (2-\gamma \alpha) \right\}^{\frac{d}{2}} \exp\left\{ -\tfrac{\gamma \alpha}{2-\gamma \alpha} \text{tr}(\bm{S}) \right\} \left| (2-\gamma \alpha)^{-1} I_{d} + \left( \tfrac{\gamma \alpha}{2-\gamma \alpha} \right)^{2} \bm{S} \right|^{-\frac{1}{d}},
\end{align*}
where $\alpha = \frac{n_{1}}{n}$ is the proportion of sample information used at the first stage and $\bm{S} = n \bs{\Sigma} (\bar{\bs{y}}_{J} - \bar{\bs{y}}_{1})(\bar{\bs{y}}_{J} - \bar{\bs{y}}_{1} )'$ quantifies the different between the posterior mean of the (first stage) partial and full posterior distributions. In this case, only the sufficient statistic $\bm{S}$ is needed to perform the optimization, which can be carried out using any root-finding algorithm obeying the constraint $\gamma \in (0, \frac{2}{\alpha})$. 

\textbf{Simulation:} Through simulation, we show that the optimal $\gamma^*$ obtained via \eqref{optimal_gamma} yields the highest ESS among all considered values of $\gamma$. Furthermore, we show that the ESS obtained from both PP-RB and SNIS estimators agree up to Monte Carlo error, corroborating the asymptotic equivalence between the two approaches established in Section \ref{snis_pprb}.

We generated samples $y_{1}, \dots, y_{n} \sim \text{Normal}(\theta, \sigma^{2})$, with $n = 200$, $\theta = 2$, and $\sigma^{2} = 5$, and assumed the Bayesian model in (\ref{bayesian_model}), with $\theta_{0} = 0$, and $\sigma^{2}_{0} = \Sigma_{0}^{-1} = 10^{4}$. Figure \ref{fig:Powered Posterior} shows the density of $[\theta \mid \bm{y} ]$, $[\theta \mid \bm{y}_1 ]$, and $[\theta \mid \bm{y}_1 ]_{\gamma}$ with with $n_{1} = 40$ and $\gamma = 0.1863$, where $\gamma$ was obtained according to (\ref{optimal_gamma}). Given the significant mean shift from partial to full posterior distribution, the $\chi^{2}$-optimal power is less than 1 so that the tail of the powered posterior better covers the region of high full posterior density. Other simulations, omitted for brevity, show that for small values of the mean shift, the optimal power can be greater than 1.

We fit the Bayesian model to the same simulated data to estimate $\theta$ using the SNIS and a two-stage PP-RB estimators. For each of the $K=80$ values of $\gamma$ equally spaced in $(0, 1.5)$, we used the $[\theta \mid \bm{y}_1]$ with $n_1 = 40$ as the proposal distribution. For each proposal, we generated $M=10,000$ samples and computed both the SNIS estimate of $\theta$ and the PP-RB estimate, where the latter was obtained using 
$100\cdot M$ MCMC iterations at the second stage.

To empirically estimate the ESS of each estimator, we independently replicated the experiment $R=300$ times for each value of $\gamma$ and estimated the variance of the resulting estimators. Figure \ref{fig:Optimal Power} shows the ESS of both estimators as a function of $\gamma$, along with the theoretical ESS obtained via the delta method. As expected from the asymptotic equivalence between PP-RB and SNIS, all three ESS estimates are in close agreement, corroborating our theoretical analysis.

\begin{figure}[t]
    \centering
    \includegraphics[width=0.95\textwidth]{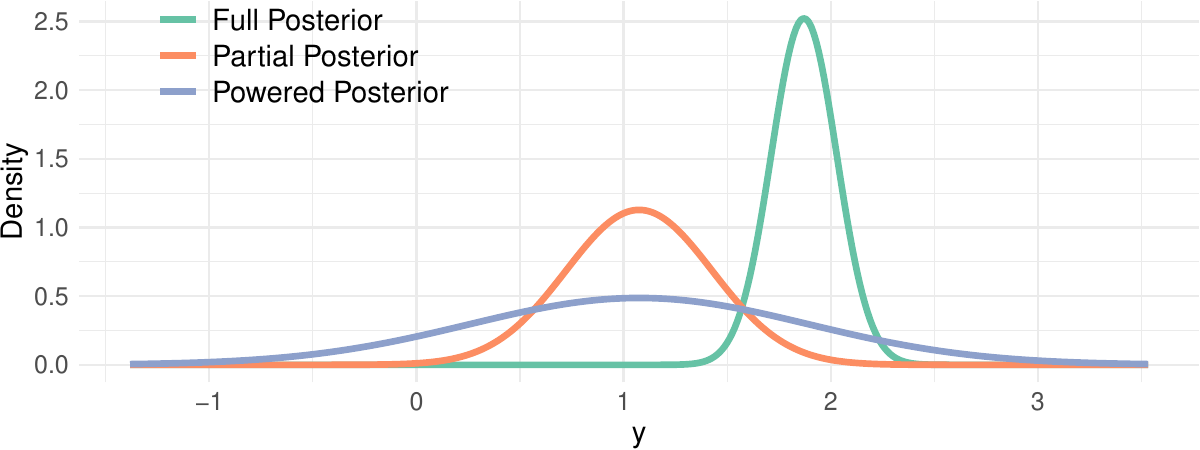}
    \caption{Comparison of the densities functions of the full posterior distribution, partial posterior distribution, and powered partial posterior distribution calculated for the simulated data.}
    \label{fig:Powered Posterior}
\end{figure}
\begin{figure}[t]
    \centering
    \includegraphics[width=0.95\textwidth]{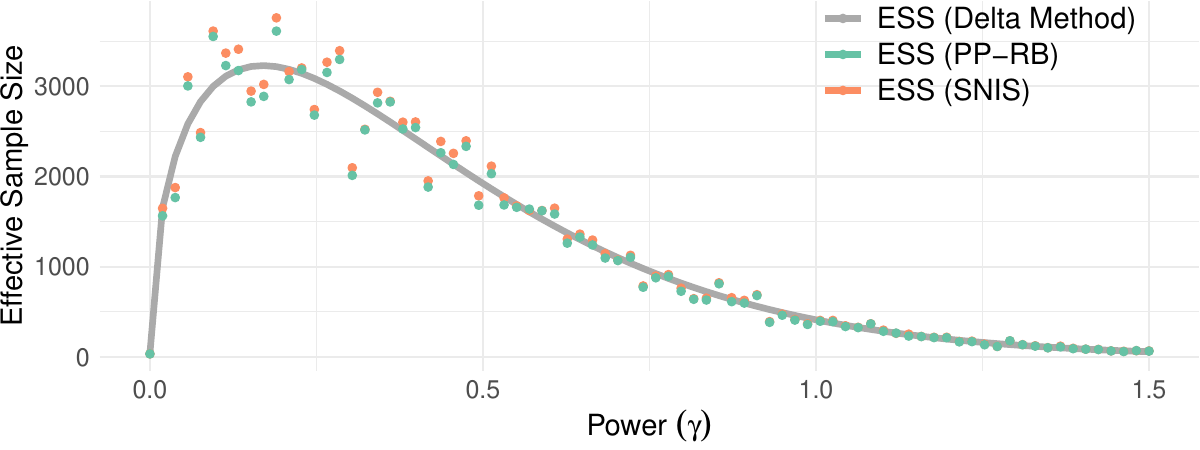}
    \caption{ESS estimated using $R = 300$ independent replicates of the SNIS and PP-RB estimators obtained using $K = 80$ different powered posterior distributions and the proposal. Solid line shows the ESS estimate obtained via the delta method.}
    \label{fig:Optimal Power}
\end{figure}

\subsubsection{Delta Method Approximation of the SNIS Estimator Variance}\label{sec:delta method}

We wish to calculate the ESS of the SNIS estimator for a given function of interest $f$, which we can write as
\begin{equation}\label{eq:ESS SNIS}
    \begin{aligned}
        \text{ESS}(f) 
        = M \frac{\text{Var}\left( \hat{I}(f) \right)}{\text{Var}\left( \tilde{I}_{\text{SNIS}}(f) \right)}
        = M \frac{ \text{Var}\left( \frac{1}{M} \sum_{m=1}^{M} f(\tilde{\bs{\theta}}_{i}) \right) }{ \text{Var}\left( \sum_{m=1}^{M} w(\bs{\theta}_{i}) f(\bs{\theta}_{i}) \right) },
    \end{aligned}
\end{equation}
where $\tilde{\bs{\theta}}_{1}, \dots, \tilde{\bs{\theta}}_{M} \overset{\text{i.i.d.}}{\sim} [\bs{\theta} | \bs{y}_{1:n}]$ (the target distribution) and $\bs{\theta}_{1}, \dots, \bs{\theta}_{M} \overset{\text{i.i.d.}}{\sim} [\bs{\theta}]_{\text{prop}}$, and
\begin{align*}
       w(\bs{\theta}) &= \frac{\tilde{w}(\bs{\theta})}{\sum_{j=1}^{M} \tilde{w}(\bs{\theta})}, \quad \text{with } \tilde{w}(\bs{\theta}) = \frac{[\bs{\theta} | \bs{y}_{1:n}]}{[\bs{\theta}]_{\text{prop}}}.
\end{align*}
Under the i.i.d. setting, the numerator in (\ref{eq:ESS SNIS}) can always be expressed as
\begin{align*}
        \text{Var}\left( \frac{1}{M} \sum_{m=1}^{M} f(\tilde{\bs{\theta}}_{i}) \right) 
        \overset{\text{ind.}}{=} \frac{1}{M^{2}} \sum_{m=1}^{M} \text{Var}\left( f(\tilde{\bs{\theta}}_{m}) \right) 
        \overset{\text{i.d.}}{=} \frac{1}{M} \text{Var}\left( f(\tilde{\bs{\theta}}_{1}) \right).
\end{align*}
However, even though the SNIS estimator also assumes i.i.d. samples, no such simple form is available for the denominator given that the SNIS weights depend on all sampled quantities. Note however, that we can rewrite \citep{kong1992note}
\begin{align*}
        \tilde{I}_{\text{SNIS}}(f) 
        &= \sum_{m=1}^{M} w(\bs{\theta}_{m}) f(\bs{\theta}_{m})
        = \frac{\frac{1}{M} \sum_{m=1}^{M} \tilde{w}(\bs{\theta}_{m}) f(\bs{\theta}_{m}) }{ \frac{1}{M} \sum_{m=1}^{M} \tilde{w}(\bs{\theta}_{m}) } \\
        &= h\left( \frac{1}{M} \sum_{m=1}^{M} \tilde{w}(\bs{\theta}_{m}) f(\bs{\theta}_{m}) , \frac{1}{M} \sum_{m=1}^{M} \tilde{w}(\bs{\theta}_{m}) \right) = h\left( \tilde{I}_{\text{IS}}(f), \tilde{I}_{\text{IS}}(1) \right),
\end{align*}
where $h(x, y) = \frac{x}{y}$ and $\tilde{I}_{\text{IS}}(f) = \frac{1}{M} \sum_{m=1}^{M} \tilde{w}(\bs{\theta}_{m}) f(\bs{\theta}_{m})$, and by the multivariate central limit theorem for i.i.d. random variables we have
\begin{align*}
        \sqrt{M} \left( \left[ \begin{array}{c} \tilde{I}_{\text{IS}}(f)  \\ \tilde{I}_{\text{IS}}(1) \end{array} \right] - \bs{\mu} \right) \overset{d}{\to} \text{Normal}_{2}\left( \bs{0}, \bs{\Sigma} \right),
\end{align*}
where 
\begin{align*}
        \bs{\mu} &= \mathbb{E}\left[ \begin{array}{c} \tilde{w}(\bs{\theta}_{1}) f(\bs{\theta}_{1}) \\ \tilde{w}(\bs{\theta}_{1})  \end{array} \right] = \left[ \begin{array}{c} I(f) \\ 1 \end{array} \right], &
        \bs{\Sigma} &= \text{Var}\left( \left[ \begin{array}{cc} \tilde{w}(\bs{\theta}_{1}) f(\bs{\theta}_{1}) \\ \tilde{w}(\bs{\theta}_{1}) \end{array} \right] \right),
\end{align*}
assuming that all entries of $\bs{\mu}$ and $\bs{\Sigma}$ exist and are finite. Applying the delta method yields
\begin{align*}
        \sqrt{M} \left(h\left( \tilde{I}_{\text{IS}}(f) , \tilde{I}_{\text{IS}}(1) \right) - h(\bs{\mu}) \right) \overset{d}{\to} \text{Normal}\left( 0, \left\{ \nabla h(\bs{\mu}) \right\}^{'} \bs{\Sigma} \left\{ \nabla h(\bs{\mu}) \right\} \right),
\end{align*}
resulting in the asymptotic approximation
\begin{align*}
        &\text{Var}\left( \tilde{I}_{\text{SNIS}}(f) \right) 
        \approx \frac{1}{M} \left\{ \nabla h(\bs{\mu}) \right\}^{'} \bs{\Sigma} \left\{ \nabla h(\bs{\mu}) \right\} 
        = \frac{1}{M} \left[ \begin{array}{c} 1 \\ -I(f) \end{array} \right]' \bs{\Sigma} \left[ \begin{array}{c} 1 \\ -I(f) \end{array} \right] \\
        &= \frac{1}{M} \bigg\{ \text{Var}\Big( \tilde{w}(\bs{\theta}_{1}) f(\bs{\theta}_{1}) \Big) - 2\{ I(f) \} \text{Cov}\Big(\tilde{w}(\bs{\theta}_{1}) f(\bs{\theta}_{1}), \tilde{w}(\bs{\theta}_{1}) \Big) + \left\{I(f)\right\}^{2} \text{Var}\Big(\tilde{w}(\bs{\theta}_{1}) \Big) \bigg\} \\
        &= \frac{1}{M} \bigg\{ \text{Var}\Big( \tilde{w}(\bs{\theta}_{1}) f(\bs{\theta}_{1}) \Big) - 2 \text{Cov}\Big( \tilde{w}(\bs{\theta}_{1}) f(\bs{\theta}_{1}), \tilde{w}(\bs{\theta}_{1}) I(f) \Big) + \text{Var}\Big( \tilde{w}(\bs{\theta}_{1}) I(f) \Big) \bigg\} \\
        &= \frac{1}{M} \text{Var}\Big( \tilde{w}(\bs{\theta}_{1}) f(\bs{\theta}_{1}) - \tilde{w}(\bs{\theta}_{1}) I(f) \Big)
        = \frac{1}{M} \text{Var}\Big( \tilde{w}(\bs{\theta}_{1}) \left\{ f(\bs{\theta}_{1}) - I(f) \right\} \Big) \\
        &= \frac{1}{M} \mathbb{E}\Big( \tilde{w}^{2}(\bs{\theta}_{1}) \left\{ f(\bs{\theta}_{1}) - I(f) \right\}^{2} \Big),
\end{align*}
yielding
\begin{align*}
        \text{ESS}(f) \approx M \frac{\text{Var}\Big(f(\tilde{\bs{\theta}}_{1})\Big)}{\mathbb{E}\Big( \tilde{w}^{2}(\bs{\theta}_{1}) \left\{ f(\bs{\theta}_{1}) - I(f) \right\}^{2} \Big)}.    
\end{align*}
\subsection{Proof of Corollaries in Section 3.2.}
Without loss of generality, we assume that the data are partitioned into two sets $\bm{y}=(\bm{y}_1',\bm{y}_2')'$ and consider a two-stage algorithm ($L=2$) for the proof. Also, we assume that we have two chains; $\ell=1$ (cold) and $\ell=2$ (hot).
\subsubsection{Corollary 3.1.}
\begin{proof}[Proof of Corollary 3.1.] The target distribution for the 1st stage is:
\begin{align*}
[\bm{\theta}|\bm{y}_1]_{\tau_{\ell}}\propto [\bm{y}_1 | \bm{\theta}]^{1/\tau_{\ell}} [\bm{\theta}],
\end{align*}
where $\tau_{\ell}=1$ and $\tau_{2}>1$ for cold and hot chain, respectively. Then, the detailed balance holds when:
\begin{align}
\label{detailed}
[\bm{y}_1|\bm{\theta}^{(k)}]^{1/\tau_{\ell}}  [\bm{\theta }^{(k)}] \cdot  P(\bm{\theta}^*| \bm{\theta}^{(k)})=[\bm{y}_1|\bm{\theta}^{*}]^{1/\tau_{\ell}} [\bm{\theta }^*]\cdot P(\bm{\theta}^{(k)}|\bm{\theta}^{*}),
\end{align}
where $\bm{\theta}^{(k)}$ and $\bm{\theta}^*$ are posterior sample at k${th}$ iteration and new proposals. We let $P(\bm{\theta}^*| \bm{\theta}^{(k)})$ present transition kernel from $\bm{\theta}^{(k)}$ to $\bm{\theta}^*$. We define the transition kernel as 
\begin{align}
\label{transition_kernel}
P(\bm{\theta}^*| \bm{\theta}^{(k)})&= Q(\bm{\theta}^*| \bm{\theta}^{(k)})  \alpha_{1,\ell} (\bm{\theta}^{(k)},\bm{\theta}^*) +  \left( 1- \int Q(\tilde{\bm{\theta}} | \bm{\theta}^{(k)}) \alpha_{1,\ell}(\bm{\theta}^{(k)}, \tilde{\bm{\theta}}) \, d\tilde{\bm{\theta}} \right) \, \delta(\bm{\theta}^* - \bm{\theta}^{(k)}),
\end{align}
where $\delta$ is Dirac delta function, $Q(\bm{\theta}^*| \bm{\theta}^{(k)})$  is the proposal distribution, and $\alpha_{1,\ell} (\bm{\theta}^{(k)},\bm{\theta}^*)=\min (1,r_{1,\ell}),$ where $r_{1,\ell}$ is 
\begin{align*}
r_{1,\ell}=\frac{Q(\bm{\theta}^{(k)}| \bm{\theta}^{*})\cdot   [\bm{y}_1 | \bm{\theta}^{*}]^{1/\tau_{\ell}} \cdot [\bm{\theta}^*]}{ Q(\bm{\theta}^{*}| \bm{\theta}^{(k)})\cdot [\bm{y}_1 | \bm{\theta}^{(k)}]^{1/\tau_{\ell}} \cdot [\bm{\theta}^{(k)}]}.
\end{align*}
By plugging \eqref{transition_kernel} into \eqref{detailed}, we consider the case $\bm{\theta}^* \neq \bm{\theta}^{(k)}$, 
for which the second term in \eqref{transition_kernel} vanishes. The case $\bm{\theta}^* = \bm{\theta}^{(k)}$, trivially satisfies the detailed balance. The LHS and RHS are then
\begin{align*}
\text{LHS}&=[\bm{y}_1|\bm{\theta}^{(k)}]^{1/\tau_{\ell}} [\bm{\theta }^{(k)}] \cdot Q(\bm{\theta}^*| \bm{\theta}^{(k)})\cdot   \min(1, r_{1,\ell} ) \\
\text{RHS}&=[\bm{y}_1|\bm{\theta}^{*}]^{1/\tau_{\ell}} [\bm{\theta }^*] \cdot Q(\bm{\theta}^{(k)}| \bm{\theta}^*) \cdot   \min(1, 1/r_{1,\ell} ).
\end{align*}
\textbf{Case 1: $r_{1,\ell} \leq 1$}
\begin{align*}
\text{LHS: }&= [\bm{y}_1|\bm{\theta}^{(k)}]^{1/\tau_{\ell}} [\bm{\theta }^{(k)}]\cdot Q(\bm{\theta}^*| \bm{\theta}^{(k)}) \cdot \frac{ Q(\bm{\theta}^{(k)}| \bm{\theta}^{*})[\bm{y}_1|\bm{\theta}^{*} ]^{1/\tau_{\ell}}  [\bm{\theta}^*]}{ Q(\bm{\theta}^*| \bm{\theta}^{(k)})[\bm{y}_1|\bm{\theta}^{(k)} ]^{1/\tau_{\ell}} [\bm{\theta}^{(k)}]} \\
&=[\bm{y}_1|\bm{\theta}^{*}]^{1/\tau_{\ell}} [\bm{\theta }^{*}]\cdot Q(\bm{\theta}^{(k)}| \bm{\theta}^{*}) \\
\text{RHS: }&=[\bm{y}_1|\bm{\theta}^{*}]^{1/\tau_{\ell}} [\bm{\theta }^{*}]\cdot Q(\bm{\theta}^{(k)}| \bm{\theta}^{*}) \cdot 1 
\end{align*}
Therefore, the detailed balance condition holds for case 1. We can show it still holds for \textbf{Case 2: $r_{1,\ell}>1$} similarly.
\end{proof}

\subsubsection{Corollary 3.2.}
\begin{proof}[Proof of Corollary 3.2.]
We first show that detailed balance holds for the within-chain update of each chain. 
For the within-chain update in second stage, the MH acceptance ratio for the $\ell$th chain at the $k$th MCMC iteration is 
\begin{align*}
r_{2,\ell}=\frac{ [\bm{y}_2| \bm{\theta}^*, \bm{y}_{1}]^{1/\tau_{\ell}}  }{  [\bm{y}_2| \bm{\theta}^{(k)}, \bm{y}_{1}]^{1/\tau_{\ell}} },
\end{align*}
where the proposal distribution is $[\bm{\theta} | \bm{y}_{1}]_{\tau_{\ell}} \propto [\bm{\theta}][\bm{y}_1|\bm{\theta}]^{1/\tau_{\ell}}$. Therefore, the transition kernel is 
\begin{align}
\label{transition_2}
P(\bm{\theta}^* | \bm{\theta}^{(k)}) = [\bm{\theta}^{*} | \bm{y}_{1}]_{\tau_{\ell}} \cdot \alpha_{2,\ell}(\bm{\theta}^{(k)},\bm{\theta}^*) 
+ \left( 1 - \int [\tilde{\bm{\theta}} | \bs{y}_{1} ]_{\tau_{\ell}} \cdot \alpha_{2,\ell}(\bm{\theta}^{(k)}, \tilde{\bm{\theta}} ) \, d\tilde{\bm{\theta}} \right) \, \delta(\bm{\theta}^* - \bm{\theta}^{(k)}),
\end{align}
where $\alpha_{2,\ell}(\bm{\theta}^{(k)},\bm{\theta}^*)=\min(1, r_{2,\ell})$. Note that the detailed balance holds when:
\begin{align}
\label{detailed2}
[\bm{y}_2,\bm{y}_1|\bm{\theta}^{(k)}]^{1/\tau_{\ell}} [\bm{\theta }^{(k)}] \cdot  P(\bm{\theta}^*| \bm{\theta}^{(k)})=[\bm{y}_2,\bm{y}_1|\bm{\theta}^{*}]^{1/\tau_{\ell}}  [\bm{\theta }^*]\cdot P(\bm{\theta}^{(k)}|\bm{\theta}^{*}).
\end{align}
By plugging \eqref{transition_2} into \eqref{detailed2}, the LHS and RHS are
\begin{align*}
\text{LHS: }&=[\bm{y}_2|\bm{y}_1, \bm{\theta}^{(k)}]^{1/\tau_{\ell}} [\bm{y}_1 | \bm{\theta}^{(k)}]^{1/\tau_{\ell}} [\bm{\theta}^{(k)}] [\bm{y}_1|\bm{\theta}^*]^{1/\tau_{\ell}}[\bm{\theta}^*] \cdot \min(1,r_{2,\ell})\\
\text{RHS: }&=[\bm{y}_2|\bm{y}_1, \bm{\theta}^{*}]^{1/\tau_{\ell}} [\bm{y}_1 | \bm{\theta}^*]^{1/\tau_{\ell}}  [\bm{\theta}^{*}]  [\bm{y}_1|\bm{\theta}^{(k)}]^{1/\tau_{\ell}} [\bm{\theta}^{(k)}] \cdot \min(1,1/r_{2,\ell}).
\end{align*}
\textbf{Case 1: $r_{2,\ell} \leq 1$}
\begin{align*}
\text{LHS: }&= [\bm{y}_2|\bm{y}_1, \bm{\theta}^{(k)}]^{1/\tau_{\ell}} [\bm{y}_1 | \bm{\theta}^{(k)}]^{1/\tau_{\ell}}  [\bm{\theta}^{(k)}]   [\bm{y}_1|\bm{\theta}^*]^{1/\tau_{\ell}} [\bm{\theta}^*]\frac{[\bm{y}_2 |\bm{\theta}^*,\bm{y}_1]^{1/\tau_{\ell}}}{[\bm{y}_2 |\bm{\theta}^{(k)},\bm{y}_1]^{1/\tau_{\ell}}}\\
&= [\bm{y}_1 | \bm{\theta}^{(k)}]^{1/\tau_{\ell}} [\bm{\theta}^{(k)}]   [\bm{y}_1|\bm{\theta}^*]^{1/\tau_{\ell}} [\bm{\theta}^*] [\bm{y}_2 |\bm{\theta}^*,\bm{y}_1]^{1/\tau_{\ell}} \\
\text{RHS: }&=[\bm{y}_2|\bm{y}_1, \bm{\theta}^{*}]^{1/\tau_{\ell}}  [\bm{y}_1 | \bm{\theta}^*]^{1/\tau_{\ell}}  [\bm{\theta}^{*}]  [\bm{y}_1|\bm{\theta}^{(k)}]^{1/\tau_{\ell}} [\bm{\theta}^{(k)}] \cdot 1
\end{align*}
\textbf{Case 2:} $r_{2,\ell} > 1$ can be checked similarly. Therefore, the detailed balance holds for the within-chain update.

To show that the detailed balance still holds for the between-chain exchange, we first define the joint tempered distribution as:
\begin{align*}
\Pi(\bm{\Theta})=\prod_{\ell=1}^{2} [\bm{\theta}|\bm{y}]_{\tau_{\ell}},
\end{align*}
where each chain targets the tempered distribution as:
\begin{align*}
[\bm{\theta}|\bm{y}]_{\tau_{\ell}} \propto [\bm{y}|\bm{\theta}]^{1/\tau_{\ell}} [\bm{\theta}].
\end{align*}
For a between-chain exchange between cold and hot chains, we define the total Markov chain as:
\begin{align*}
\bm{\Theta}= (\bm{\theta}_{\text{cold}} , \bm{\theta}_{\text{hot}}),    
\end{align*}
where $\bm{\theta}_{\text{cold}}$ and $\bm{\theta}_{\text{hot}}$ are parameters for cold and hot chains, respectively. Then, consider
\begin{align*}
\bm{\Theta}&= (\bm{\theta}_{\text{cold}},\bm{\theta}_{\text{hot}})\\
\bm{\Theta}^*&= (\bm{\theta}_{\text{hot}},\bm{\theta}_{\text{cold}}),
\end{align*}
where $\bm{\theta}_{\text{cold}}$ and $\bm{\theta}_{\text{hot}}$ are swapped. The transition kernel is given by: 
\begin{align}
\label{transition3}
P(\bm{\Theta}^*|\bm{\Theta}) = Q(\bm{\Theta}^*|\bm{\Theta})\,\alpha_{swap}(\bm{\Theta},\bm{\Theta}^*) 
+ \left( 1 - \int Q(\tilde{\bm{\Theta}} | \bm{\Theta}) \, \alpha_{swap}(\bm{\Theta}, \tilde{\bm{\Theta}}) \, d\tilde{\bm{\Theta}} \right)\,\delta(\bm{\Theta}^*-\bm{\Theta}),
\end{align}
where $Q$ is symmetric, $\alpha_{swap}(\bm{\Theta},\bm{\Theta^*})=\text{min}(1,R)$, and $\alpha_{swap}(\bm{\Theta^*},\bm{\Theta})=\text{min}(1,1/R)$. $R$ is given by
\begin{align*}
R&=\frac{\Pi(\bm{\Theta}^*)}{\Pi(\bm{\Theta})}.
\end{align*}
The detailed balance condition requires 
\begin{align}
\label{detailed3}
\Pi(\bm{\Theta})P(\bm{\Theta}^*|\bm{\Theta})=\Pi(\bm{\Theta}^*)P(\bm{\Theta}|\bm{\Theta}^*)
\end{align}
By plugging \eqref{transition3} into \eqref{detailed3}, the LHS and RHS are
\begin{align*}
\text{LHS} =\Pi(\bm{\Theta})\alpha_{swap}(\bm{\Theta},\bm{\Theta^*}) \\
\text{RHS} =\Pi(\bm{\Theta})\alpha_{swap}(\bm{\Theta^*},\bm{\Theta}) 
\end{align*}
\textbf{Case 1: }$R>1$, then $\alpha_{swap}(\bm{\Theta}^*,\bm{\Theta})=1$ and 
\begin{align*}
\text{LHS: }&\Pi(\bm{\Theta}) \cdot 1 \\
\text{RHS: }&\Pi(\bm{\Theta}^*) \cdot \frac{1}{R}=\Pi (\bm{\Theta})
\end{align*}
Similarly, it holds for case 2: $R\leq 1$. Therefore, the detailed balance condition holds for the between-chain update. 
\end{proof}

\subsubsection{Corollary 3.3.}
\begin{proof}[Proof of Corollary 3.3.]
Marginalizing out the hot chain gives
\begin{align*}
\int \mathbf{\Pi}(\bm{\Theta}) d\bm{\theta}_2
&= [\bm{\theta}|\bm{y}]_{\tau_1}\int [\bm{\theta}|\bm{y}]_{\tau_{2}} d\bm{\theta}_2\\
&=[\bm{\theta}|\bm{y}]_{\tau_1}\\
&= [\bm{\theta}|\bm{y}]^{1/\tau_1} [\bm{\theta}]\\
&\propto [\bm{\theta}|\bm{y}] \quad (\because \tau_1=1).
\end{align*}    
\end{proof}

\subsection{1989 Loma Prieta Earthquake analysis with two partitions}
\begin{figure}[h]
\centering
\includegraphics[width=1\linewidth]{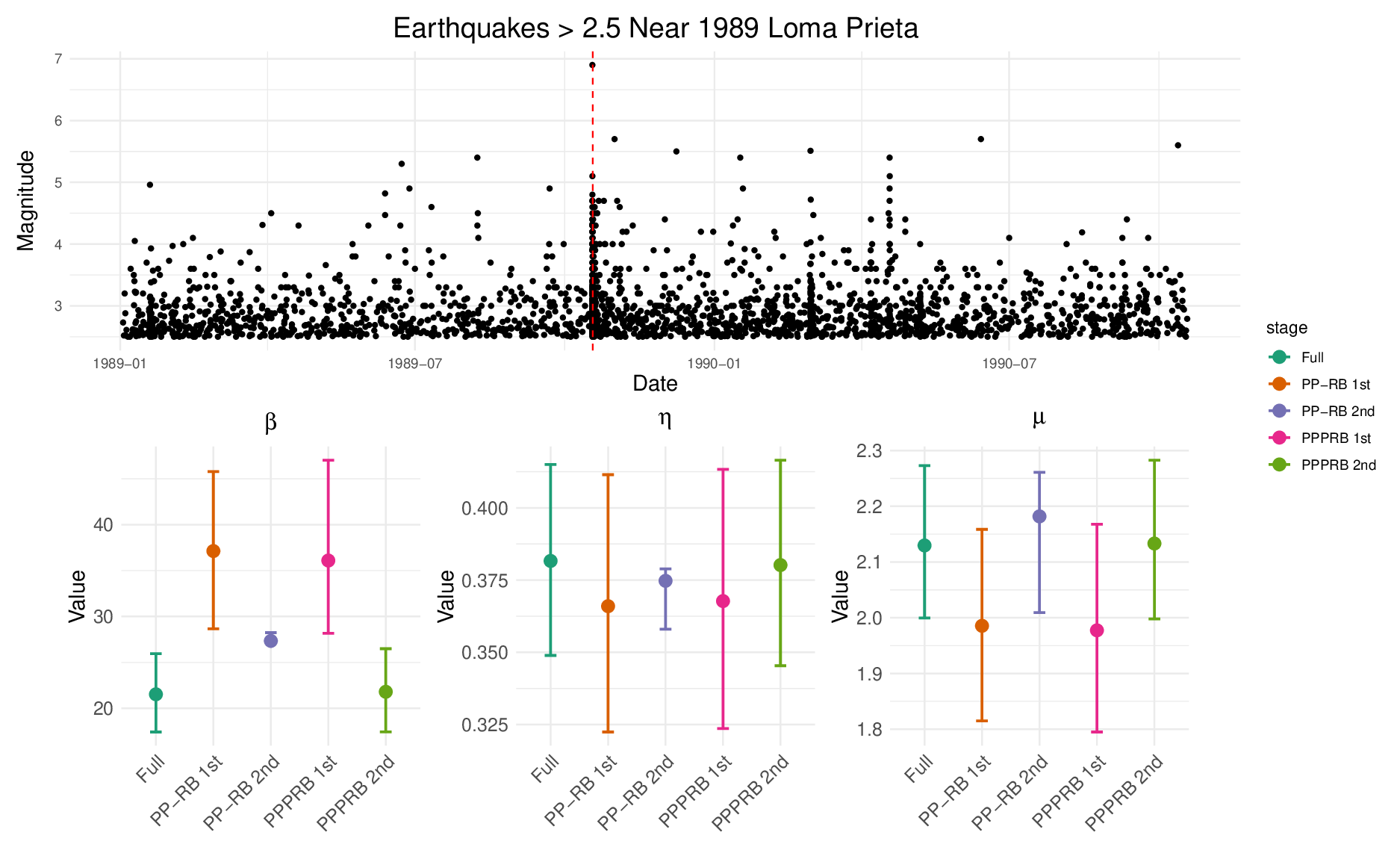}
\caption{Top: Time series plot of earthquake data. Orange dashed vertical lines indicate the time points at which the data are partitioned. Bottom left, bottom middle and bottom right: Posterior summaries comparison for standard MCMC (full data), PP-RB, and PPP-RB, respectively.}
\label{fig_hawkes2}
\end{figure}
We compare the PP-RB and PPP-RB using the Earthquake dataset, where the data are partitioned into two subsets based on the occurrence of the 1989 Loma Prieta earthquake (Figure \ref{fig_hawkes2}). In this case, the posterior distribution obtained by PP-RB differs substantially from that obtained by the MH algorithm. In contrast, the posterior distribution obtained by PPP-RB remains close to that of the MH algorithm, demonstrating its robustness.

\end{document}